\documentclass[10pt,journal,a4paper]{IEEEtran}
\usepackage{amssymb,amsmath}
\usepackage{cite}
\usepackage{graphicx,subfigure}
\usepackage{psfrag}
\usepackage{url}
\usepackage[latin1]{inputenc}
\usepackage[absolute,overlay]{textpos}

\usepackage{pgf}
\usepackage[latin1]{inputenc}
\usepackage{algorithm}
\usepackage{algpseudocode}

\usepackage{amsthm}

\newtheorem{lemma}{Lemma}
\newtheorem{theorem}{Theorem}

\begin{document}

\title{Rate Control for Wireless-Powered Communication Network with Reliability and Delay Constraints}
\author{
	\IEEEauthorblockN{	Onel L. Alcaraz López, Student Member, IEEE, 
		Hirley Alves, Member, IEEE, 
		Richard Demo Souza, Senior Member, IEEE,
		Samuel Montejo-Sánchez, Member, IEEE, and
		Evelio M. García Fernández, Member, IEEE
	}
	\thanks{O.L.A. López, H. Alves are with the Centre for Wireless Communications (CWC), University of Oulu, Finland. \{onel.alcarazlopez,hirley.alves\}@oulu.fi.}
	\thanks{R.D. Souza is with Federal University of Santa Catarina (UFSC), Florianópolis, Brazil. \{richard.demo@ufsc.br\}.}	
	\thanks{S. Montejo-Sánchez is with Programa Institucional de Fomento a la I+D+i, Universidad Tecnológica Metropolitana, Santiago, Chile. \{smontejo@utem.cl\}.}	
	\thanks{E.M.G. Fernández is with Federal University of Paraná (UFPR), Curitiba, Brazil. \{evelio@ufpr.br\}.}
	\thanks{This work is partially supported by Academy of Finland (Aka) (Grants n.303532, n.307492, n.318927 (6Genesis Flagship)), CAPES (Brazilian Agency for Higher Education) (project PrInt CAPES-UFSC ``Automation 4.0''), as well as FONDECYT Postdoctoral Grant n.3170021.} 	
}					
					
\maketitle

\begin{abstract}	
	 We consider a two-phase Wireless-Powered Communication Network under Nakagami-m fading, where a wireless energy transfer process first powers a sensor node that then uses such energy to transmit its data in the wireless information transmission phase. 	 
	 We explore a fixed transmit rate scheme designed to cope with the reliability and delay constraints of the system while attaining closed-form approximations for the optimum wireless energy transfer and wireless information transmission blocklength. Then, a more-elaborate rate control strategy exploiting the readily available battery charge information is proposed and the results evidence its outstanding performance  when compared with the fixed transmit rate, for which no battery charge information is available. It even reaches an average rate performance close to that of an ideal scheme requiring full Channel State Information at transmitter side. 
	  Numerical results show the positive impact of a greater number of antennas at the destination, and evidence that the greater the reliability constraints, the smaller the message sizes on average, and the smaller the optimum information blocklengths. Finally, we corroborate the appropriateness of using the asymptotic blocklength formulation as an approximation of the non-asymptotic finite blocklength results.	
\end{abstract}
\section{Introduction}\label{int}
Very recently, energy harvesting (EH) has emerged as a promising technology to achieve a sustained and low-cost operation of low-power devices such as sensors or tiny actuators. Apart from the conventional EH sources \cite{Sudevalayam.2011} such as solar, wind, kinetic, or other ambient energy source, the radio-frequency (RF) signals are very attractive because they can carry energy and information simultaneously \cite{Tandon2.2016}, which enables energy constrained nodes to harvest energy and receive information. Notice that RF EH, also known as Wireless Energy Transfer (WET), is envisioned to be of utmost importance in scenarios where replacing or recharging the batteries may be dangerous, as in a toxic environment, or even impossible, as in sensors implanted in human bodies. Those scenarios will be also present in future communication paradigms such as the Internet of Things (IoT) \cite{Chu.2018}, where powering a potentially massive number of devices will be a major challenge.

Wireless Powered Communication Networks (WPCN) is already a well-addressed topic in the literature. Readers can find in \cite{Lu.2015} an overview of WPCNs including system architecture, WET techniques, and existing applications. 
As shown in \cite{Ju.2014} for the case of throughput maximization, WPCN can even achieve a performance comparable to that of conventional (non-WET) networks when intelligent policies are applied. Some of the main strategies that have been considered in the scientific literature are relay-assisted \cite{Krikidis.2012,Ye.2019,Chen.2015,Krikidis.2014}, Automatic Repeat-reQuest (ARQ) \cite{Witt.2014,Mao.2016}, and power control \cite{Liu.2013,Chu.2018,Liu.2017,Isikman.2017,KhanYazdan.2017,Wu.2018} mechanisms.
Specifically, the work in \cite{Krikidis.2012} deals with RF energy transfer in a three-node cooperative network, and focuses on the optimum switching between EH and data relaying, while the outage probability of a dual-hop decode-and-forward (DF) relay system utilizing the power-splitting EH protocol is investigated in \cite{Ye.2019}. Authors in \cite{Chen.2015} also consider a cooperative WPCN but this time the source and relay nodes have no embedded energy supply and are wirelessly powered by an access point (AP).  The performance of a large-scale WPCN with random number of transmitter-receiver pairs and potential relays that are randomly distributed into the network, is characterized in \cite{Krikidis.2014} by using
stochastic-geometry tools. In the latter two cases authors show the benefits of cooperation in terms of throughput and outage probability, respectively. An investigation on the usage of the time-switching protocol when hybrid ARQ schemes are used for information transmission is presented in \cite{Witt.2014}, while results demonstrate the improvement of the system performance for high Signal-to-Noise Ratio (SNR). Additionally, a novel ARQ protocol for EH receivers is proposed in \cite{Mao.2016} by allowing the ACK feedback to be adapted based upon the receiver's EH state. Therein, authors derive optimal reception policies including the sampling, decoding and feedback strategies, for different ARQ protocols and assuming several power consumption sources. 
Regarding power control, authors in \cite{Liu.2013} study a WPCN that employs  dynamic power splitting. Under a point-to-point flat-fading channel setup, they resort to the optimal power splitting rule at the receiver based on the  Channel State Information (CSI) to optimize the rate-energy
performance trade-off. Meanwhile, \cite{Liu.2017} investigates the optimal energy beamforming and time assignment in WPCN for smart cities, and \cite{Isikman.2017} proposes a low-complexity solution, called fixed threshold transmission (FTT) by setting a transmit power threshold to determine whether transmission takes place or not. On the other hand, the energy efficiency-optimal downlink transmit power for a massive multiple-input multiple-output (MIMO)  base station (BS) using WET to charge single-antenna EH users, is derived in \cite{KhanYazdan.2017}. Therein, authors use a scalable model for the devices' circuit power consumption and show that it is indeed energy efficient to operate the system in the massive MIMO regime. Finally, a time-division multiple access (TDMA)-based WPCN was shown in \cite{Wu.2018} to be more spectral-efficient than its non-orthogonal multiple access (NOMA) counterpart after optimizing the transmit power and time allocation in both cases. Specifically, the scenario consists of a power beacon (PB) powering wirelessly a set of devices, which in turn use such harvested energy for feeding their circuits and transmitting information to a receiving AP, while all the devices are single-antenna.

Aforementioned research works, and many others related, design and/or evaluate strategies to maximize the probability of transmission, the  throughput, or the energy efficiency, or to minimize the transmitted power, or the outage probability. However, as enabler of IoT paradigm, WPCNs find a natural application in Machine-Type Communication (MTC) use cases of coming 5th Generation (5G) of wireless systems, where some stringent requirements in reliability and delay may be present, as is the case of ultra-reliable MTC (uMTC) scenarios \cite{Durisi.2016}.
Critical connections for industrial automation, smart grids, reliable wireless coordination
among vehicles, and others, are some examples of uMTC scenarios \cite{Popovski.2014} with general delay constraints less than 10ms and error probability requirements not greater than $10^{-3}$. The interplay between reliability/delay makes physical layer design of uMTC very complex and challenging \cite{Hyoungju.2017}, and WPCN systems need to be designed properly to meet such demands. 
\subsection{Related Works}\label{A}
Some recent works \cite{Lopez2.2017,Haghifam.2016,Lopez3.2017,Lopez4.2018,Khan.2017,Lopez.2018,Tandon.2016,Dabirnia.2016,Makki.2016} deal with those issues by considering short packets (stringent delay constraints) and/or ultra-reliability. In \cite{Lopez2.2017} we analyze and optimize a single-hop wireless system with energy transfer in the downlink and information transfer in the uplink, under quasi-static Nakagami-m fading in uMTC scenarios. The results demonstrate that there
is an optimum number of channel uses for both energy and information transfer for a given message length. 
Cooperative WPCNs are considered in \cite{Haghifam.2016,Lopez3.2017,Lopez4.2018}, where the impact of an amplify-and-forward (AF) \cite{Haghifam.2016} and DF \cite{Lopez3.2017} relay-assisted communication setup without direct link
is evaluated, and  power consumption sources beyond data transmission power, as well as imperfect CSI and a direct link, are incorporated in \cite{Lopez4.2018} for the DF setup. 
A save-then-transmit protocol is investigated in \cite{Khan.2017}  where authors address the scenario in which one or more wireless PBs use WET to charge a node over fading channels, and then it attempts to communicate with another receiver over a noisy channel and using a finite number of channel uses. The system performance is characterized using metrics, such as the energy supply probability at the transmitter, and the achievable rate at the receiver.
In \cite{Lopez.2018} we study transmission strategies for scenarios with/without energy accumulation between transmission rounds, where a power control protocol for the energy accumulation scenario is proposed in order to improve the reliability performance under the given delay constraints. On the other
hand, subblock energy-constrained codes are investigated in \cite{Tandon.2016}, and a sufficient condition on the subblock length to avoid energy outage at the receiver is provided. Meanwhile, authors in \cite{Dabirnia.2016} propose using concatenation of a nonlinear trellis code (NLTC) with an outer low-density parity-check (LDPC) code  taking into account both energy transmission and error rate requirements, and show that the designed codes operate at $\sim\! 0.8$ dB away from the information theoretic limits.
Besides, retransmission protocols, in both energy and information transmission phases, are implemented in \cite{Makki.2016} to reduce the outage probability compared to open-loop communication.  Nevertheless, none of these works deal with rate allocation strategies for WPCNs under reliability and delay constraints. Notice that in WPCNs with \textit{elastic} applications, in the sense that devices can transmit data over a wide range of data rates, the transmission rate is a degree of freedom that should be properly exploited.
\subsection{Contributions and organization of this paper}\label{B}
Herein we aim at filling the above gap in the literature by addressing the rate control problem in such WPCNs with reliability and delay constraints. Specifically, our goal is to maximize the instantaneous message size such that requirements of latency and reliability are met.
For a two-phase wireless system, where first a WET process powers a low-power sensor node that uses that energy to transmit its data in the Wireless Information Transmission (WIT) phase, we propose two rate control strategies: i) fixed transmit rate (FTR) and ii) transmit rate with known state of charge (KSC). We compare both with the ideal scheme that requires full CSI at transmitter side (fCSI). The main contributions of this work can be summarized as follows:
\begin{itemize}
	\item We show the advantages of deciding on the transmit rate based on the battery charge, since it allows 	considerable rate improvements when compared to 	the fixed rate scheme, achieving a performance close to that of the ideal scheme where full CSI is also available at the transmitter side;	
	\item Closed-form approximations for the attainable average message size of the proposed schemes  are obtained, and also for the optimum WET and WIT blocklength when the fixed transmit rate scheme is used;
    \item Results show the positive impact of a greater number of antennas at the destination, and evidence that the more stringent the reliability constraints, the smaller the message sizes on average, and the smaller the optimum WIT blocklength. The inverse relation between the optimum WIT blocklength and the reliability constraints was proved analytically for the fixed transmit scheme, but numerical results evidence that this characteristic holds also for the other schemes;    
    \item We develop the asymptotic formulation for the rate control problem under the proposed schemes. We show how to depart from these expressions to attain non-asymptotic results, while illustrating numerically the gaps of both formulations. We corroborate the appropriateness of using the asymptotic formulation as an approximation of the non-asymptotic finite blocklength results in this scenario, specially at data rates that are not low.    
\end{itemize}

Next, Section \ref{system} presents the system model and assumptions. Sections \ref{FTR}, \ref{KSC} and \ref{fCSI} discuss the rate control schemes, FTR, KSC and fCSI, respectively, designed to meet the given reliability and delay constraints. Section \ref{results} presents the numerical results, while Section \ref{conclusions} concludes the paper.
\begin{table}[!t]
	\caption{Main Symbols}
	\label{table}
	\centering
	\begin{tabular}{|p{0.9cm}p{6.9cm}|}
		\hline
		$v,n,\delta$ & WET channel uses, WIT channel uses, $\delta=n+v$ \\
		$v^*,n^*$ & Optimum values of $v$ and $n$, respectively \\
		$T_c$ & Duration of a channel use\\
		$\lambda_{xy}$ & Path loss between nodes $X$ and $Y$\\
		$h$ &  Normalized channel power gain of the link $T\rightarrow S$\\
		$g_i$ & Normalized channel power gain of the link between $S$ and   the $i$-th antenna of $D$\\
		$m_1,m_2$ & Nakagami-m shape factor of link $T\rightarrow S$ and $S\rightarrow D$, respectively\\
		$M$ & Number of antennas at $D$\\
		$\textsl{g}$ & RV denoting the normalized sum of $g_i$ for $i=1,...,M$\\
		$E$ & Energy harvested at $S$ during the WET phase\\
		$P_x$ & Transmit power of node $X$\\
		$\eta$ & Energy conversion efficiency\\
		$k$ & Message size in bits\\
		$k_0$ & Fixed minimum message size\\
		$r$ & Transmit rate, $r=k/n$\\
		$\epsilon_{\mathrm{th}}$	& Maximum allowed error probability\\
		$\sigma_{d}^2$ & Average noise power at $D$ \\
		$\gamma_i$ & Instantaneous SNR at the $i$-th antenna of $D$\\
		$\gamma$ & Instantaneous SNR after Maximum Ratio Combining (MRC)  at $D$\\
		$\psi$ & Average SNR at $D$ normalized by $M$ and the quotient $v/n$\\
		$w$ & RV denoting the product $h\textsl{g}$\\
		$\epsilon(\gamma,k,n)$ & Error probability for an Additive White Gaussian Noise (AWGN) channel\\
		$h_0,z_0$& Channel thresholds for KSC and fCSI schemes, respectively\\		
		$\epsilon_{\mathrm{th}}^*$ & Average error probability when operating above the channel  thresholds, e.g., above $h_0$ for KSC or above $z_0$ for fCSI\\
		\hline		
	\end{tabular}
\end{table}
\textbf{Notation:} $X\sim\Gamma(m,1/m)$ is a normalized gamma distributed random variable (RV) with shape factor $m$, Probability Density Function (PDF) $f_X(x)=\frac{m^m}{\Gamma(m)}x^{m-1}e^{-mx}$ and Cumulative Distribution Function (CDF) $F_X(x)=1-\frac{\Gamma(m,mx)}{\Gamma(m)}$ for $x\ge 0$. $\mathbb{P}[A]$ is the probability of event $A$, while $\mathbb{E}_z[\!\ \cdot\ \!]$ denotes expectation with respect to RV $Z$. $\operatorname{K}_t(\mathrel{\cdot})$ is the modified Bessel function of second kind and order $t$,  $\mathrm{E}_1(\mathrel{\cdot})$ is the exponential integral \cite{Gradshteyn.2014}, while $\mathcal{W}(\mathrel{\cdot})$ is the main branch of the Lambert W function \cite{Corless.1996}, which satisfies $\mathcal{W}(x)\ge-1$ for $x\in \mathcal{R}$ and it is defined in $-\tfrac{1}{e}\!\le\! x\!<\!0$. The Gaussian Q-function is denoted
as $Q(x)=\int_{x}^{\infty}\frac{1}{\sqrt{2\pi}}e^{-t^2/2}\mathrm{d}t$ and $\sup\{\cdot\}$ is the supremum operation. 
Table~\ref{table} summarizes the main symbols used throughout this paper.
\section{System Model}\label{system}
Consider the two-phase scenario in Fig.~\ref{Fig1}. First, $T$ powers $S$ during $v$ channel uses in the WET phase. Then, $S$, which may be a sensor node with very limited energy supply, uses the energy obtained in the WET phase to transmit its information to $D$ over $n$ channel uses in the WIT phase. We do not consider other power consumption sources at $S$, e.g., as sensing operations and/or receive/processing/transmit circuitry functioning as considered in \cite{Mao.2016,KhanYazdan.2017,Wu.2018,Liu.2017,Lopez4.2018}. However, notice that our derivations keep accurate in case that such additional consumption is small compared to the average energy harvested by $S$ as assumed implicitly or explicitly in \cite{Chu.2018,Ju.2014,Krikidis.2012,Ye.2019,Chen.2015,Krikidis.2014,Witt.2014,Liu.2013,Lopez2.2017,Haghifam.2016,Lopez3.2017,Khan.2017,Lopez.2018,Makki.2016,Boshkovska.2015}. Otherwise, modeling and including into the analysis other power consumption sources may be strictly necessary, as failing to satisfy the corresponding energy requirements could become the main cause of system failure, as shown for instance in \cite{Lopez4.2018} when transmitting short data messages in a cooperative scenario. Meanwhile, notice that in systems operating under stringent reliability and delay constraints, which are the focus of this work, such energy-outage events could be critical and difficult to predict/avoid, hence, deploying a very small battery, but still able of supporting the basic operation of the EH devices, could be required. In such case the harvested energy would support only the transmission operations, while once again our derivations and analysis hold. On the other hand, we assume that both, energy and data transfer, utilize the same frequency spectrum, which guarantees small form factors, but perhaps more important it allows saving spectrum and enables Simultaneous Wireless and Information Transfer (SWIPT) \cite{Krikidis.2014,Makki.2016} in the link $T\rightarrow S$, e.g., for the transmission of control messages from $T$ to $S$. 
We denote by $T_c$ the duration of a channel use, thus, the duration of an entire transmission round is $\delta T_c$, where $\delta=v+n$. 
Also, $S$ and $T$ are single antenna devices while $D$ is equipped with $M$ antennas.
The path loss between $T$ and $S$, and between $S$ and each antenna of $D$, is denoted as $\lambda_{ts}$, $\lambda_{sd}$, respectively.

\begin{figure}[t!]
	\centering
	\includegraphics[width=0.47\textwidth]{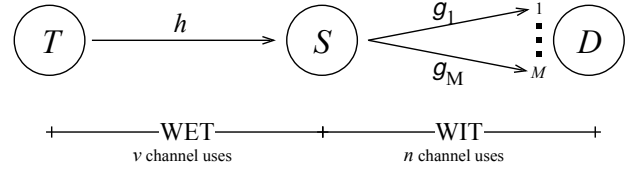}
	\caption{System model.}		
	\label{Fig1}	
\end{figure}
Nakagami-m quasi-static channels are assumed where the fading is considered to be constant and independent and identically distributed over each phase.
We consider normalized channel power gains $h\sim\Gamma(m_1,\tfrac{1}{m_1})$ and $g_i\sim\Gamma(m_2,\tfrac{1}{m_2})$, where $h$ characterizes the link $T\rightarrow S$, $g_i$ is the coefficient of the link between $S$ and the $i-$th antenna of $D$, and $m_1,m_2$ are their corresponding Nakagami-m Line-Of-Sight (LOS) factors\footnote{Usually $m_1>m_2$ since WET processes require greater LOS than WIT.}. Also, the $M$ antennas at $D$ are sufficiently separated such that $g_i,\ i=1,\cdots,M$ are independent, and spatial diversity can be fully achieved.

A typical application scenario could be that in which $T$ is a BS who powers a set of nearby low power devices including $S$, while enabling their transmissions to a sink or aggregating node $D$ with $M$ antennas. Notice that our analyses considering a single device $S$ remain valid even in a multi-user setup if the simultaneous transmissions in the WIT phase occur in different frequency sub-bands\footnote{In such case, the MRC scheme (and consequent derivations) considered in Subsection~\ref{snrD} remains valid as it can be applied after filtering to each sub-band to maximize its associated SNR}. On the other hand, $T$ could also be a  dedicated RF source deployed to provide a predictable energy supply to those devices, or an interrogator requesting information from them, for which $D=T$ and $\lambda_{ts}=\lambda_{sd}$. The latter case implies that instead of a single-antenna device, $T$ would be equipped  with $M$ antennas. In that sense we would like to point out that our derivations remain valid in setups where $T$ is equipped with $N$ antennas in two different scenarios: 
\begin{itemize}
	\item when $T$ uses Maximum Ratio Transmission (MRT) to serve a single device $S$. In such case, our derivations hold by utilizing $Nm_1$ and $NP_t$ instead of $m_1$ and $P_t$, respectively;	
	\item when $T$ uses the CSI-free Switching Antennas (SA) strategy proposed in \cite{Lopez5.2018} under which $T$ transmits with full power through one antenna at a time such that all antennas are used during a coherence block and channels are uncorrelated. Assuming equal-time allocation for each antenna the performance is equivalent to the one attained under the assumption of a single antenna at $T$ but utilizing $m_1N$ instead of $m_1$.
\end{itemize}
Notice that the SA scheme is particularly suitable for low-power multi-user setups as it does not require CSI acquisition neither other signaling procedures, and still is capable of providing an $N-$fold diversity gain in the energy supplying.
\subsection{WET Phase}\label{wet}
During each WET phase the energy harvested at $S$ is given by
\begin{align}\label{Eh}
E=\frac{\eta P_{t}h}{\lambda_{ts}} vT_c,
\end{align}
where $P_{t}$ is the transmit power of $T$ and sufficiently large such that the energy harvested from noise is negligible, while $0\!<\!\eta\!<\!1$ is the energy conversion efficiency. Notice that we are considering the simple linear EH model as in \cite{Chu.2018,Ju.2014,Krikidis.2012,Ye.2019,Chen.2015,Krikidis.2014,Witt.2014,Liu.2013,Liu.2017,Lopez2.2017,Haghifam.2016,Lopez3.2017,Lopez4.2018,Khan.2017,Makki.2016} to allow some analytical tractability and facilitate the discussions. Although the specific performance results must vary when utilizing different EH models, the functioning of the proposed rate control schemes does not change, while their relative performance is expected to be similar. In any case, the analysis of the proposed schemes under more evolved/practical EH models, such as the piecewise linear model \cite{Lopez.2018,Lopez5.2018} and/or other non-linear models \cite{Boshkovska.2015,Chen.2017}, is left for a future work. 
\subsection{WIT Phase}\label{WIT}
After the WET phase, $S$ uses the harvested energy to transmit its message $\iota\in\{1,\cdots,L\}$ using a $(n,L,P_s,\epsilon_{\mathrm{th}})$ code, which is defined as the collection of 
\begin{itemize}
	\item an encoder $\Upsilon: \{1,\cdots,L\}\mapsto \mathcal{C}^n$, which maps the message $\iota$ into a length-$n$ codeword $x_{\iota}\in\{x_1,\cdots,x_L\}$ satisfying the power constraint
	 \begin{align}\label{xj}
	 \frac{1}{n}||x_j||^2\le P_s,\ \forall j,
	 \end{align}
	 where $P_s$ is given by
	 \begin{equation}\label{Ps}
	 P_{s}=\frac{E}{nT_c}=\frac{v}{n}\frac{\eta P_t h}{\lambda_{ts}},
	 \end{equation}
	 assuming  that $S$ uses all the harvested energy in \eqref{Eh} for transmission;
	 \item a decoder $\Lambda: \mathcal{C}^n\mapsto \{1,\cdots,L\}$ satisfying the maximum error probability constraint
	 \begin{align}\label{const}
	 \mathbb{E}\big[\max\limits_{\forall j}\mathbb{P}\big[\Lambda(y)\ne J|J=j\big]\big]\le \epsilon_{\mathrm{th}},
	 \end{align}
	 where $y$ denotes the channel output induced by the transmitted codeword at the end of each WIT phase, e.g., the received signal vector, and the expectation is taken over the channel realizations $h$ and $g_i,\ \forall i$. 
\end{itemize}
The maximum message size is then
\begin{align}\label{k}
 k=\sup\{\log_2 L: \exists (n,L,P_s,\epsilon_{\mathrm{th}})\ \mathrm{code} \}\ \mathrm{(bits)}.
\end{align}
Notice that the value of $k$ in \eqref{k} is largely influenced by the channel conditions as $y$ in \eqref{const} depends on them. Therefore, we next focus on characterizing the SNR at the destination.
\subsection{SNR at $D$}\label{snrD}
The instantaneous SNR at antenna $i$ of $D$ is
	\begin{equation}\label{SNR}
	\gamma_i=\frac{P_{s}g_i}{\lambda_{sd}\sigma_{d}^2}=\frac{v}{n}\psi h g_i,
	\end{equation}
	where $P_s$ is given in \eqref{Ps}, $\sigma_{d}^2$ is the average noise power at $D$, and
	\begin{align}\label{psi}
	\psi=\frac{\eta P_t}{\lambda_{ts}\lambda_{sd}\sigma_{d}^2}
	\end{align}
	is the average SNR at $D$ normalized by the number of antennas, $M$, and the quotient of the WET and WIT blocklengths, $v/n$. 
	Notice that for a fixed $\delta$, by increasing $v$ such that $S$ collects more energy, $S$ can increase its transmit power and consequently the SNR at each antenna of $D$ improves, but that leads to a decrease in $n$, and therefore an increase in the coding rate $r=k/n$ which degrades the reliability of the system. This is a well known trade-off of WPCNs. 
	
	Finally, perfect CSI at $D$ is assumed in the decoding after the WIT phase\footnote{CSI acquisition in an energy-limited setup is not trivial and including the effect of imperfect CSI and associated costs would demand a more elaborated mathematical analysis that is out of the scope of this work. However, notice that when channels remain constant over multiple transmission rounds, which should be the case in most scenarios when operating with short blocklengths, the cost of CSI acquisition can be negligible \cite{Lopez2.2017}.}, and $D$ uses MRC to combine the arriving information at each antenna. Therefore, the overall SNR seen is
	\begin{align}
	\gamma&=\sum_{i=1}^{M}\gamma_i=\frac{v}{n}\psi h\sum_{i=1}^{M}g_i=\frac{v}{n}M\psi h\textsl{g} \stackrel{(a)}{=}\frac{v}{n}M\psi w, \label{gam}
	\end{align}
	where $\textsl{g}=\frac{1}{M}\sum_{i=1}^{M}g_i\sim\Gamma(m_2M,\tfrac{1}{m_2M})$.
	Notice that $(a)$ comes from defining $w=h\textsl{g}$, whose PDF and CDF are given in the following result.
	\begin{lemma}\label{lem_1}
		The PDF and CDF of $W$ are given by
		\begin{align}
		f_W(w)&=\frac{2(m_1m_2M)^{\frac{m_2M+m_1}{2}}w^{\frac{m_2M+m_1}{2}-1}}{\Gamma(m_2M)\Gamma(m_1)}\times\nonumber\\
		&\qquad\qquad\qquad\times\operatorname{K}_{m_2M-m_1}\big(2\sqrt{m_1m_2Mw}\big),\label{fW}\\
		F_W(w)&=\frac{4\big(m_1m_2Mw)^{\frac{m_2M+m_1}{2}}}{\Gamma(m_2M)\Gamma(m_1)}\int\limits_{0}^{1}x^{m_2M+m_1-1}\times\nonumber\\
		&\qquad\qquad\times\operatorname{K}_{m_2M-m_1}\big(2\sqrt{m_1m_2Mw}x\big)\mathrm{d}x, \label{FW}\\
		&=1-\frac{2(Mw)^{\frac{M}{2}}}{(M-1)!}\operatorname{K}_M(2\sqrt{Mw}),\label{FW1}
		\end{align}
		for $w\ge 0$, where \eqref{FW1} holds only for Rayleigh fading scenarios, e.g., $m_1=m_2=1$.
	\end{lemma}
	\begin{proof}
		See Appendix~\ref{App_A}. \phantom\qedhere
	\end{proof}	
\subsection{Reliability-Rate Trade-off}\label{C2} 
As commented in Subsection~\ref{WIT},
	$k$, $n$ and $\epsilon$ are related according to \eqref{k}. Recently, Polyanskiy \textit{et. al} \cite{Polyanskiy.2010} have  provided an accurate characterization of the trade-off between these parameters in AWGN channels and for $n\ge 100$ channel uses. By explicitly including the effect of the SNR at the receiver, the trade-off is given by
	\begin{align}\label{ke}
	k\approx nC(\gamma)-\sqrt{nV(\gamma)}Q^{-1}(\epsilon_{\mathrm{th}})\log_2e,
	\end{align}
where $C(\gamma)=\log_2(1+\gamma)$ is the Shannon capacity and $V(\gamma)=1-\frac{1}{(1+\gamma)^2}$ is the channel dispersion, which measures the stochastic variability of the channel relative to a deterministic channel with the same capacity. Now, we can rewrite \eqref{ke} as
\begin{align}
\epsilon(\gamma,k,n)\approx Q\Biggl(\frac{C(\gamma)-k/n}{\sqrt{V(\gamma)/n}}\ln 2\Biggl),\label{ep}
\end{align}
which is the maximum error probability when transmitting $k$ information bits over a channel with SNR $\gamma$ and using $n$ complex symbols. 
Notice that \eqref{ep} matches the asymptotic \textit{outage probability} when $n\rightarrow \infty$ and/or $\gamma\rightarrow 0$. Meanwhile, for quasi-static fading channels the average maximum error probability is \cite[Eq.(59)]{Yang.2014}
\begin{align}
\bar{\epsilon}(k,n)\approx\mathbb{E}_{\gamma}\Biggl[Q\Biggl(\frac{C(\gamma)-k/n}{\sqrt{V(\gamma)/n}}\ln 2\Biggl)\Biggl], \label{EX}
\end{align}
since channel becomes conditionally Gaussian on $\gamma$, and we only require
to take expectation over that RV to attain the corresponding average error probability. However, and as it has been shown in \cite{Mary.2016}, the effect of the fading when evaluating \eqref{EX} makes disappear the impact of the finite blocklength when $k$ is not extremely small and there is not a strong LOS component, thus, the asymptotic outage probability, which is the  Laplace approximation of \eqref{EX}, is a good match in those cases and is given by 
\begin{align}
\bar{\epsilon}(k,n)&\approx\mathbb{P}\big[\gamma<2^{k/n}-1\big]=F_\gamma(2^{k/n}-1).\label{inf}
\end{align}
\subsection{Problem Statement}\label{PS}
Our goal in this work is finding the maximum instantaneous message size $k$ such that $\bar{\epsilon}(k,n)\le \epsilon_{\mathrm{th}}$ for given $v, n$, thus, fixed delay $\delta$.\footnote{Notice that $nT_c$ and $vT_c$ are the duration of the WIT and WET phases, and in most scenarios these durations must be fixed, e.g., for simpler synchronization and/or time-slotted operation.} We constrain $k$ to be above some fixed value $k_0$ in each transmission, which represents the minimum possible message size. Such value is determined by the minimum frame size in the physical layer, which includes the mandatory metadata fields in the particular frame structure, hence it would be impossible communicating with a message size smaller than $k_0$. Notice that in our scenario we do not perform reliability analysis of metadata and payload separately and still assume a fixed $k_0$. In such case our results hold by assuming that metadata and payload are jointly encoded, which favors the system reliability \cite{Popovski.2014}.

We refer to this as the rate control problem since by finding $k$ we set the transmission rate as $k/n$, and we address it for three different schemes, when		
\begin{itemize}
	\item $S$ adopts a fixed transmit rate (Section~\ref{FTR}) independently of the randomness of the system parameters, e.g., fading realizations;
	\item $S$ uses the information on the remaining battery charge (Section~\ref{KSC}), which is intrinsically related to the $T\rightarrow S$ WET channel realization, to decide on its rate;
	\item $S$ uses the full knowledge of all, WET and WIT channel realizations, to decide on its rate (Section~\ref{fCSI}), which is idealistic for a practical implementation and it is used as benchmark.
\end{itemize}

We assume that $S$ is aware of the number of antennas at $D$ and the value of $\psi$ in  \eqref{psi}. Notice that $\lambda_{sd}$ and $\sigma_{d}^2$ in \eqref{psi} can be acquired, with low frequency, by $S$ via a low-rate  feedback channel from $D$, also the remaining terms there can be estimated by knowing the average energy harvested at $S$, $\bar{E}=\eta P_{t}v T_c/\lambda_{ts}$. 

Finally, we adopt the asymptotic expression  of $\bar{\epsilon}(k,n)$ in \eqref{inf} as an initial guess or lower-bound for gaining in analytical tractability. This is, we start with the asymptotic formulation and then we state the procedure to obtain the non-asymptotic results based on the former. Additionally, in Section~\ref{results} we investigate the gap in performance between both, asymptotic and non-asymptotic, formulations. In the entire work we focus on the region where $\epsilon_{\mathrm{th}}< 10^{-1}$.
\section{Fixed Transmit Rate (FTR)}\label{FTR}
Herein, we analyze a rate allocation scheme which does not rely on any instantaneous information at $S$, hence, $S$ adopts a fixed rate $r=k/n$.
Notice that for being possible to meet the required reliability constraint it is necessary that $\bar{\epsilon}(k_0,n)\le\epsilon_{\mathrm{th}}$ since $\bar{\epsilon}(k,n)$ is an increasing function of $k$ and $k\ge k_0$. 
The allocated message size comes from solving equation $\bar{\epsilon}(k,n)=\epsilon_{\mathrm{th}}$ for $k$,  which can only be performed numerically when using \eqref{EX}. 
Thus, we use the approximation in \eqref{inf} as follows
\begin{align}
\bar{\epsilon}(k,n)&\approx\mathbb{P}\big[\gamma<2^{k/n}-1\big]\nonumber\\
&\stackrel{(a)}{=}\mathbb{P}\Big[\frac{v}{n}M\psi w<2^{k/n}-1\Big]\nonumber\\
&\stackrel{(b)}{=}F_{W}\Big(\frac{\big(2^{k/n}-1\big)}{ M\psi}\frac{n}{v}\Big),\label{eqe}
\end{align}
where $(a)$ comes from using \eqref{gam} and $(b)$ from the definition of the CDF. Now, the maximum message size comes from setting $\bar{\epsilon}(k,n)=\epsilon_{\mathrm{th}}$ in \eqref{eqe} and isolating $k$, which yields
\begin{align}
k_{_\mathrm{FTR}}&\!\approx\! n\log_2\!\Big(1\!+\!\frac{v}{ n}M\psi F_{W}^{-1}(\epsilon_{\mathrm{th}})\Big)\!=\!n\log_2\!\Big(1\!+\!\frac{v}{n}\chi\Big),
 \label{k1}
\end{align}
where $\chi=M\psi F_W^{-1}(\epsilon_{\mathrm{th}})$. In case that $k_{_\mathrm{FTR}}<k_0$ the FTR scheme can not satisfy the reliability requirements at the same time that satisfying the message size constraint, thus, the rate allocation is infeasible and no transmission occurs.
Notice that using \eqref{k1} requires finding the inverse of \eqref{FW}, which is not possible in closed form. Following result addresses that issue for the reliability region of our concern.
\begin{theorem}\label{the_1}
	$F_W^{-1}(\epsilon_{\mathrm{th}})$ converges to
	\begin{align}	
	-\mathcal{W}\left(\!\!-\frac{\Big(\frac{\Gamma(m_2M)\Gamma(m_1)\min(m_2M,m_1)}{\Gamma\big(|m_2M-m_1|\big)}\epsilon_{\mathrm{th}}\Big)^{\frac{1}{\min(m_2M,m_1)}}}{m_1m_2M}\!\right) \label{k1G}
	\end{align}
	as $\epsilon_{\mathrm{th}}\rightarrow 0$ for $m_2M\ne m_1$. For $\epsilon_{\mathrm{th}}\le 10^{-1}$, \eqref{k1G} can be used already as an accurate approximation of $F^{-1}_W(\epsilon_{\mathrm{th}})$, specially for relatively large $M$. For instance, when $m_1=4,\ m_2=2$, the relative error from using \eqref{k1G} instead of the numerically-evaluated exact inverse function of \eqref{FW} is $20\%$ and $5\%$, for $M=1$ and $M=8$, respectively.
\end{theorem}
\begin{proof}
	See Appendix~\ref{App_B}. \phantom\qedhere
\end{proof}

For setups with $m_2M=m_1$ we were not able of finding an accurate approximation, but notice that this case corresponds to very specific scenarios and its importance would be mainly theoretical rather than practical.

Returning to \eqref{k1}, notice that $\frac{v}{n}\chi$ works as an equivalent SNR which increases linearly with $\psi$. Additionally, it increases with $m_2M$ and $m_1$, and following a power-law of $\epsilon_{\mathrm{th}}$, inversely proportional to the minimum between $m_2M$ and $m_1$. The latter statement comes from using the definition of $\chi$ along with \eqref{k1G}.

As stated in Subsection~\ref{PS} our aim is to find the maximum instantaneous message size for a given $\epsilon_{\mathrm{th}},\ v$ and $n$, while the relevance of fixing $v$ and $n$ was also highlighted (see Comment~4). However, in more flexible scenarios where choosing $v$ and $n$ under the delay constraint $\delta=v+n$ is allowed, there is an optimum blocklength for WET and WIT phases, which we further investigate next.
\begin{theorem}\label{the_2}
	Taking advantage of \eqref{k1}, the optimum blocklength for WIT and WET phases is approximately given by
	\begin{align}
	n^*&\approx \frac{\chi}{\frac{\chi-1}{\mathcal{W}\big(\frac{\chi-1}{e}\big)}+\chi-1}\delta,\label{noptFTR}\\
	v^*&\approx\delta-n^*, \label{v}
	\end{align}
	respectively, and the approximation includes taking the nearest integer.
\end{theorem}
\begin{proof}
	See Appendix~\ref{App_C}. \phantom\qedhere
\end{proof}

\begin{figure}[t!]
	\centering
	\includegraphics[width=0.47\textwidth]{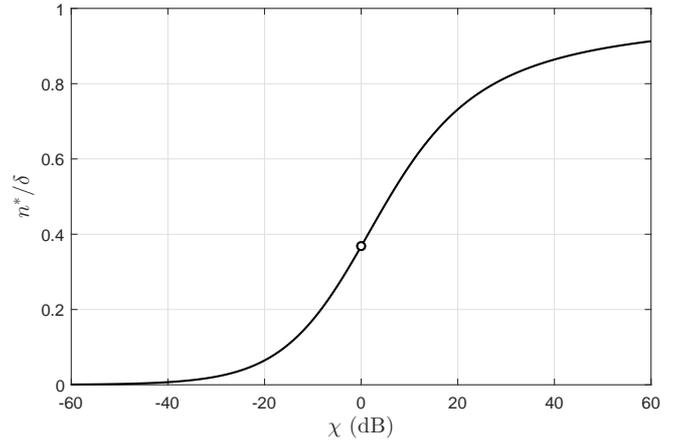}
	\caption{$n^*/\delta$ as a function of $\chi$.}		
	\label{Fig2}
\end{figure}
From \eqref{noptFTR} notice that	$n^*$ increases linearly with $\delta$, while it is also an increasing function of $\chi$ as shown in Fig.~\ref{Fig2}. Although \eqref{noptFTR} is indeterminate for $\chi=1\ (0$ dB$)$ we know that $\lim_{\chi\rightarrow 1} \frac{\chi -1}{\mathcal{W}\big(\frac{\chi-1}{e}\big)}=e$, thus, $\lim_{\chi\rightarrow 1}n^*/\delta=1/e $. Since $\chi$ is also an increasing function of $\epsilon_{\mathrm{th}}$ (because $F_W^{-1}(\epsilon_{\mathrm{th}})$ increases with $\epsilon_{\mathrm{th}}$ by definition), then, $n^*$ is an increasing function of $\epsilon_{\mathrm{th}}$.
This result is very interesting since it clearly states that the harvested energy, through the $v$ WET channel uses, becomes more and more relevant when targeting greater reliabilities than the information data rate, $k/n$. Thus, it is better improving the harvested energy statistics while transmitting with greater information data rates (since delay is fixed) for meeting more stringent reliability constraints. Finally, notice that $n^*$ does not depend on the fading characteristics when using the asymptotic formulation as so far.

Until this point we have been using the asymptotic approximation \eqref{inf} in order to attain analytical and insightful results: the message size \eqref{k1} and the optimum blocklengths \eqref{noptFTR}, \eqref{v} for a given delay $\delta$. Notice that at finite blocklength we can use the asymptotic value provided by \eqref{k1} as an initial step to solve $\bar{\epsilon}(k,n)=\epsilon_{\mathrm{th}}$ with \eqref{EX}. The procedure is illustrated in \textbf{Algorithm~\ref{alg11}}, while
\begin{algorithm}[t!]
	\caption{Finding $k$ such that $\bar{\epsilon}(k,n)=\epsilon_{\mathrm{th}}$}
	\label{alg11}
\begin{algorithmic} [1]
\State Calculate $k$ according to \eqref{k1}  \label{line01} 
\If{$k\ge k_0$}
\State Evaluate \eqref{EX} \label{line02} 
\If{$\epsilon>\epsilon_{\mathrm{th}}$} \label{line04}    
\State Decrease $k$
\State Return to line \ref{line02}						\label{line06}
\Else
\If{$k<k_0$}
\State \textit{Infeasible rate}
\State $k=0$
\EndIf
\EndIf        \label{line07}                     
\Else
\State \textit{Infeasible rate}
\State $k=0$
\EndIf
\State Return $k$
\State End \label{line08}         
\end{algorithmic}
\end{algorithm}
in Section~\ref{results} numerical results evaluate the gap between the infinite and finite blocklength formulations and further discussion follows.

In the next section we propose a rate allocation scheme that forces $S$ to exploit its battery state of charge information for performance improvements.
\section{Transmit rate with Known State of Charge (KSC)}\label{KSC}
In the previous section $S$ had to rely on a conservative fixed transmit rate since no instantaneous CSI was available. Herein we investigate the optimum variable transmit rate if $S$ exploits the information related with the state of charge of its battery, which is strictly related with the CSI of the link $T\rightarrow S$. Different from the FTR scheme where the transmit rate is fixed along the transmission rounds, in this case the required transmit rate is computed after concluding each WET phase. In that case the asymptotic outage probability is given by
\begin{align}
\bar{\epsilon}(k,n)&\approx \mathbb{P}\big[\gamma<2^{k/n}-1\big|E \big]\nonumber\\
&\stackrel{(a)}{=}\mathbb{P}\Big[\frac{v}{n}M\psi h\textsl{g}<2^{k/n}-1\Big|h\Big]\nonumber\\
&\stackrel{(b)}{=} F_{\mathcal{G}}\bigg(\frac{\big(2^{k/n}-1\big)}{M\psi h}\frac{n}{v}\bigg|h\bigg)\nonumber\\
\bar{\epsilon}(k,n,h)&\stackrel{(c)}{\approx} F_{\mathcal{G}}\bigg(\frac{\big(2^{k/n}-1\big)}{M\psi h}\frac{n}{v}\bigg),\label{eqq}
\end{align}
where $(a)$ follows directly after using \eqref{gam}, while $(b)$ comes from isolating $\textsl{g}$ and using the definition of CDF. Notice that this is different than the result in \eqref{eqe} because herein the battery charge is known, thus $h$ is not random anymore and the only RV is $\textsl{g}$.
Finally, $(c)$ comes from incorporating $h$ as an input parameter in the error function.

Obviously, for some values of $h$, $h<h_0$, we cannot find $k\ge k_0$ such that  $\bar{\epsilon}(k,n,h)=\epsilon_{\mathrm{th}}$. This is because the WET channel was so poor such that for $k\ge k_0$ we have $\bar{\epsilon}(k,n,h)>\epsilon_{\mathrm{th}}$. Therefore, the transmission strategy under this scheme is
\begin{enumerate}
	\item If $h\le h_0$, where $h_0$ works as a threshold, $S$ transmits with $k=k_0$; 
	\item otherwise; $S$ adopts a rate which is dependent on the specific $h$ such that on average it performs with error probability $\epsilon_{\mathrm{th}}^*\le \epsilon_{\mathrm{th}}$. This allows compensating the average error probability to meet the reliability requirement.
\end{enumerate}
Thus, when $k_0>0$ it is required that
\begin{align}
\int\limits_{0}^{h_0}\bar{\epsilon}(k_0,n,h)f_H(h)\mathrm{d}h+\epsilon_{\mathrm{th}}^*\int\limits_{h_0}^{\infty}f_H(h)\mathrm{d}h&\!=\!\epsilon_{\mathrm{th}}\nonumber\\
\int\limits_{0}^{h_0}\bar{\epsilon}(k_0,n,h)f_H(h)\mathrm{d}h+\epsilon_{\mathrm{th}}^*\big(1-F_H(h_0)\big)&\!=\!\epsilon_{\mathrm{th}}\nonumber\\
\underbrace{\int\limits_{0}^{h_0}\!\frac{m_1^{m_1}}{\Gamma(m_1)}\bar{\epsilon}(k_0,n,h)h^{m_1\!-\!1}e^{-m_1h}\mathrm{d}h}_{z_1(h_0)}\!+\!\underbrace{\epsilon_{\mathrm{th}}^*\frac{\Gamma(m_1,m_1h_0)}{\Gamma(m_1)}}_{z_2(h_0)}&\!=\!\epsilon_{\mathrm{th}}.\label{eq}
\end{align}
Also, above scheme implies that 
\begin{align}
\epsilon_{\mathrm{th}}^*=\bar{\epsilon}(k_0,n,h_0),\label{eq1}
\end{align}
and notice that 
\begin{align}
\lim\limits_{h_0\rightarrow \infty}\!(z_1(h_0)\!+\!z_2(h_0))&\!=\!\int\limits_{0}^{\infty}\!\frac{m_1^{m_1}}{\Gamma(m_1)}\bar{\epsilon}(k_0,n,h)h^{m_1\!-\!1}e^{-m_1h}\mathrm{d}h\nonumber\\
&=\bar{\epsilon}(k_0,n);\nonumber
\end{align}
while that is required to be equal or smaller than the target $\epsilon_{\mathrm{th}}$ for feasibility. This same result was attained in the previous section for the FTR scheme. Thus, neither KSC, nor other transmission rate allocation strategy  operating with the same values of $n$ and $v$, is capable of providing a rate that meets the average reliability constraint when $\bar{\epsilon}(k_0,n)>\epsilon_{\mathrm{th}}$.\footnote{However larger data rates are provided when the feasibility condition holds, which is illustrated in Section~\ref{results}.}

Now we are able of providing the following result characterizing the rate allocation under the KSC scheme when $\bar{\epsilon}(k_0,n)<\epsilon_{\mathrm{th}}$.
\begin{theorem}\label{the3}
	Under the KSC scheme and with $k_0>0$, $S$ chooses its message size according to
	\begin{equation}\label{kexp}
	k_{_\mathrm{KSC}}=\left\{ \begin{array}{lc}
	k_0, &\ \mathrm{if}\   h\le h_0 \\
	n\log_2\Big(1+\frac{2^{k_0/n}-1}{h_0}h\Big), &\ \mathrm{if}\ h>h_0
	\end{array}
	\right.,
	\end{equation}
	where $h_0$ is the unique solution of \eqref{eq} and can be obtained numerically by a root-finding algorithm\footnote{See also Appendix~\ref{App_F} for an specific algorithm to solve \eqref{eq}.}.
\end{theorem}
\begin{proof}
	See Appendix~\ref{App_F}. \phantom\qedhere
\end{proof}

Then, based on \eqref{kexp} the average message size is given by
\begin{align}
\bar{k}_{_\mathrm{KSC}}&\!=\!k_0\int\limits_{0}^{h_0}f_{H}(h)\mathrm{d}h\!+\!n\int\limits_{h_0}^{\infty}\!\log_2\Big(1\!+\!\frac{2^{k_0/n}\!-\!1}{h_0}h\Big)f_{H}(h)\mathrm{d}h\nonumber\\
&\!\stackrel{(a)}{=}\!k_0\Big(\!1\!-\!\frac{\Gamma(m_1,\!m_1h_0)}{\Gamma(m_1)}\!\Big)\!+\!\frac{m_1^{m_1}ne^{\!-\!m_1h_0}h_0^{m_1}}{\ln(2)\big(2^{\frac{k_0}{n}}\!-\!1\big)^{m_1}\Gamma(m_1)}\times\nonumber\\
&\ \ \ \times\int\limits_{0}^{\infty}\!\!\ln\!\Big(2^{\frac{k_0}{n}}\!+\!x\Big)\big(x\!+\!2^{\frac{k_0}{n}}\!-\!1\big)^{m_1\!-\!1}e^{\frac{\!-\!m_1h_0}{2^{k_0/n}\!-\!1}x}\!\mathrm{d}x,\label{kk}
\end{align}
where $(a)$ comes from using the CDF and PDF of $h$, and changing variable $h\rightarrow \frac{(x+2^{k_0/n}-1)h_0}{2^{k_0/n}-1}$. Unfortunately there is not a closed-form solution for the integral\footnote{There are numerous numerical methods that could be used to solve efficiently \eqref{kk}, for instance, importance sampling by taking samples from an exponential distribution with mean $\frac{2^{k_0/n}-1}{m_1 h_0}$, or the Gaussian-Chebyshev quadrature method as in \cite{Ye.2019}. In our specific setup the latter would require some kind of transformation to make the integrating interval to be finite, e.g., using $x=\tan\omega$, $0\le\omega\le\pi/2$.} unless for the specific case of Rayleigh fading, for which
\begin{align}\label{mk}
\bar{k}_{_\mathrm{KSC}}&\!=\!k_0\big(1\!-\!e^{-h_0}\big)\!+\!\frac{ne^{\!-h_0}h_0}{2^{\frac{k_0}{n}}\!-\!1}\!\!\int\limits_{0}^{\infty}\!\!\log_2\!\Big(2^{\frac{k_0}{n}}\!+\!x\Big)e^{\frac{\!-h_0}{2^{k_0/n}\!-\!1}x}\!\mathrm{d}x\nonumber\\
&\!\stackrel{(a)}{=}\!k_0(1\!-\!e^{-h_0})\!+\!n\Bigg[\frac{e^{\frac{h_0}{2^{k_0/n}\!-\!1}}}{\ln(2)}\mathrm{E}_1\Big(\frac{2^{\frac{k_0}{n}}h_0}{2^{\frac{k_0}{n}}\!-\!1}\Big)\!+\!e^{-h_0}\frac{k_0}{n}\Bigg],	
\end{align}
where $(a)$ comes  with the help of \cite[Eq.(4.337.1)]{Gradshteyn.2014} and some algebraic manipulations.
Even when attaining a closed-form expression in \eqref{mk} for the average message size under Rayleigh fading, it becomes intractable  finding the optimum value of $n$ under this scheme, mainly because $h_0$, which is the solution of \eqref{eq}, depends also on $n$. However, either evaluating numerically \eqref{kk}  or finding numerically $n^*$ in \eqref{kk} or \eqref{mk} is still useful since the Monte Carlo computation alternative could be practically infeasible for some system parameter values. This is because $\bar{k}_{_\mathrm{KSC}}$ depends on $h_0$,  and using Monte Carlo to find accurately $h_0$ requires enormous computation resources when $\epsilon_{\mathrm{th}}$ is very small, e.g., $\epsilon_{\mathrm{th}}<10^{-6}$.

Next, we analyze the procedure for attaining the non-asymptotic results, and explain how to depart from our above results to reach them.
\subsection{Finite Blocklength formulation}
Under the finite blocklength formulation $h_0$ is greater than in the asymptotic case, thus, the value of $h_0$ found after solving \eqref{eq} can be used as a starting point to solve the following more-elaborate equation  
\begin{align}
\epsilon_{\mathrm{th}}=&\int_{0}^{\infty}\biggl[\int_{0}^{h_0}\epsilon(\gamma(h,\textsl{g}),k_0,n)f_H(h)\mathrm{d}h+\nonumber\\
&\qquad+\epsilon(\gamma(h_0,\textsl{g}),k_0,n)\big(1-F_H(h_0)\big)\biggl]f_{\mathcal{G}}(\textsl{g})\mathrm{d}\textsl{g},\label{eqs}
\end{align}
which is attained in an analogous way to \eqref{eq}, but this time we require to use \eqref{ep} with $\gamma(h,\textsl{g})=M\psi h \textsl{g}\tfrac{v}{n}$ according to \eqref{gam}.\footnote{Differently to the case when solving \eqref{eq}, here it is not recommended using a Newton-like method in order to solve \eqref{eqs} since computing the derivative of the function takes more computational resources than the function itself. The fact that the solution of \eqref{eqs} is close to the solution of \eqref{eq} should be exploited.}
 After finding $h_0$ we can evaluate
\begin{align}
\epsilon_{\mathrm{th}}^*=\int_{0}^{\infty}\epsilon(\gamma(h_0,\textsl{g}),k_0,n)f_{\mathcal{G}}(\textsc{g})\mathrm{d}\textsl{g}.\label{eqep}
\end{align}

The procedure for allocating the rate under the finite blocklength formulation is described in  \textbf{Algorithm~\ref{alg12}}, and in line~\ref{ll3} notice that  $k\in\Big[k_0,n\log_2\big(1+\tfrac{2^{k_0/n}-1}{h_0}h\big)\Big)$  according to \eqref{kexp}, thus we can use an iterative procedure as in the previous section, starting from $k=k_{_{\mathrm{KSC}}}$ and decreasing $k$ at each iteration until \eqref{eqN} is met.\footnote{From a practical and feasible perspective,  $S$ should have tabulated many channel realizations, $h$, in order to decide on its rate based on the minimum of the interval where the current realization lies in.}
\begin{algorithm}[t!]
		\caption{Finding $k$ such that $\bar{\epsilon}(k,n)=\epsilon_{\mathrm{th}}$}
		\label{alg12}
\begin{algorithmic} [1]	
	\State Find $h_0$ and $\epsilon_{\mathrm{th}}^*$ according to \eqref{eqs} and \eqref{eqep}, respectively.
	\State For each round, if $h\le h_0$ then $k=k_0$; otherwise	
	\State Find the maximum integer $k$ that satisfies \label{ll3}
	\begin{align}
	\int_{0}^{\infty}\epsilon(\gamma(h,\textsl{g}),k,n)f_{\mathcal{G}}(\textsl{g})\mathrm{d}\textsl{g}\le\epsilon_{\mathrm{th}}^*\label{eqN}
	\end{align}	
	\State End \label{line18}        
\end{algorithmic}
\end{algorithm}

In the next section we analyze the case where, in addition to the battery state of charge information, $S$ knows the CSI of the channels in the link $S\rightarrow D$ and uses also that information for allocating its transmit rate.
\section{Transmit rate with full CSI at $S$ (fCSI)}\label{fCSI}
In this setup $S$ knows all the channels $h$ and $g_i, \forall i$, thus it knows what would be the value of $\gamma$, and it can adjust its transmit rate to meet the required reliability. 
Therefore, the channel is AWGN and the maximum message size comes from using \eqref{ke} and is given by\footnote{For the FTR and KSC schemes we have departed from the asymptotic analysis since they provide a valid approximation that can be used as a starting point for finding their non-asymptotic equivalent. For the fCSI case this is no longer necessary since \eqref{rate} is non-asymptotic and it is already in closed-form.}
\begin{align}
k&\approx nC(\gamma)-\sqrt{n V(\gamma)}Q^{-1}(\epsilon_{\mathrm{th}}^*)\log_2 e.\label{rate}
\end{align}
Obviously the impact of finite blocklength is not negligible here. Notice that for sufficiently long blocklength, $n\rightarrow \infty$, yields $r\rightarrow C(\gamma)$, thus, following Shannon, it is possible to transmit with an arbitrary small (not bounded) error probability with such rate.

Similar to the previous section where $h$ was known, now $S$ will transmit with an adjustable $k>k_0$ when $w=h\textsl{g}>w_0$, and $w_0$ can be found by solving
\begin{align}\label{eqF}
\int\limits_{0}^{w_0}\!\!\!\epsilon(\gamma(w),k_0,n)f_W(w)\mathrm{d}w+\epsilon_{\mathrm{th}}^*\int\limits_{w_0}^{\infty}f_W(w)\mathrm{d}w&=\epsilon_{\mathrm{th}}\nonumber\\
\int\limits_{0}^{w_0}\epsilon(\gamma(w),k_0,n)f_W(w)\mathrm{d}w+\epsilon_{\mathrm{th}}^*\big(1-F_W(w_0)\big)&=\epsilon_{\mathrm{th}},
\end{align}
where $F_W(w)$ is given in \eqref{FW}, and
\begin{align}
\epsilon_{\mathrm{th}}^*=\epsilon(\gamma(w_0),k_0,n)
\end{align}
with $\gamma(w)=M\psi w\tfrac{v}{n}$ as in \eqref{gam}. Herein we do not discuss any particular method for solving \eqref{eqF}. It does not worth our attention since the fCSI method does not seem practically feasible, because the full knowledge of all the channels at $S$ is a very strong assumption, particularly in power-limited networks; and it is only useful to provide benchmark results, therefore we resort to standard numerical solutions available in MatLab, e.g., \texttt{fsolve}, \texttt{fzero}, \texttt{vpasolve}, and Wolfram Mathematica, e.g., \texttt{FindRoot}.  

Finally, based on \eqref{rate}, whenever $w>w_0$, $S$ transmits with
\begin{align}
k_{_\mathrm{fCSI}}=nC(\gamma(w))-\sqrt{V(\gamma(w))n}Q^{-1}(\epsilon_{\mathrm{th}}^*)\log_2e,
\end{align}
while
\begin{align}
\bar{k}_{_\mathrm{fCSI}}
&=k_0\big(1-F_W(w_0)\big)+\int_{w_0}^{\infty}k_{_\mathrm{fCSI}}f_W(w)\mathrm{d}w.
\end{align}

Now, we have concluded the analytical discussions around the three rate allocation schemes and in the following section we provide some numerical performance comparisons.
\section{Numerical Results}\label{results}
In this section, we present numerical results to investigate the performance of the proposed schemes as a function of the system parameters. Unless stated otherwise, results are obtained by using the values shown in Table~\ref{table2}. Notice that $\psi=0$ dB may correspond to a system with the following practical values (similar to those in \cite{Krikidis.2014,Chen.2015,Lopez2.2017,Lopez.2018,Lopez4.2018,Lu.2015}): $\eta=0.3$, which is feasible according to the state-of-the-art in EH circuit design, $P_t=10$W, $\sigma_{d}^2=-90$dBm, log-distance path loss model with $\lambda_{ts}=\kappa d_{ts}^{\alpha}$ and $\lambda_{sd}=\kappa d_{sd}^{\alpha}$, where $\kappa=30$dB is the average signal power attenuation at a reference distance of 1 meter and accounts for other factors such as the carrier frequency and heights and gains of the antennas, $d_{ts}=10$m and $d_{sd}=14$m are the lengths of the links $T\rightarrow S$ and $S\rightarrow D$, and $\alpha=3$ is the path loss exponent. While by setting $m_1=5$ and $m_2=2$ we model both WIT and WET links with certain LOS as required in practical WPCN setups. Also, by taking into account that the typical message size in uMTC use cases is 32 bytes \cite{3GPP} while the value of $k_0$ must be much smaller we have set $k_0 = 2$ bytes. Finally, notice that the overall delay $\delta=v+n=1200$ channel uses\footnote{Since $\delta$ must be small for short-latency transmissions, both $v$ and $n$ are small as well (strictly smaller than $\delta$) but $v$ has usually a stronger impact on the system performance than $n$ according to the analysis around Fig.~\ref{Fig2}, which is also illustrated in \cite{Lopez2.2017,Lopez.2018} for different setups, hence, its value should be set larger than $n$. As shown in Table~\ref{table2} we utilize $n=200,\ v=1000$ channel uses.} corresponds to $3.6$ms by assuming $T_c=3 \mu$s as in \cite{Lopez2.2017}.
\begin{table}[!t]
	\centering
	\caption{System parameters}
	\label{table2}
	\begin{tabular}{cc}
		\hline		
		\textbf{Parameter}	& \textbf{Value}	 \\ \hline
		$k_0$				& $2$ bytes (16 bits)		\\
		$m_1,\ m_2$			& $5$, $2$ (similar to \cite{Lopez2.2017,Lopez.2018})	\\
		$n,\ v$			    & $200,\ 1000$ channel uses (similar to \cite{Lopez.2018})	\\
		$\psi$				& $0$ dB ~($\!\!$\cite{Chen.2015,Lopez2.2017,Lopez.2018,Lopez4.2018,Krikidis.2014,Lu.2015})	\\
		$T_c$			& $3\ \mu$s \cite{Lopez2.2017} 	\\	
		\hline
	\end{tabular}
\end{table}
\begin{figure}[t!]
	\centering
	\includegraphics[width=0.47\textwidth]{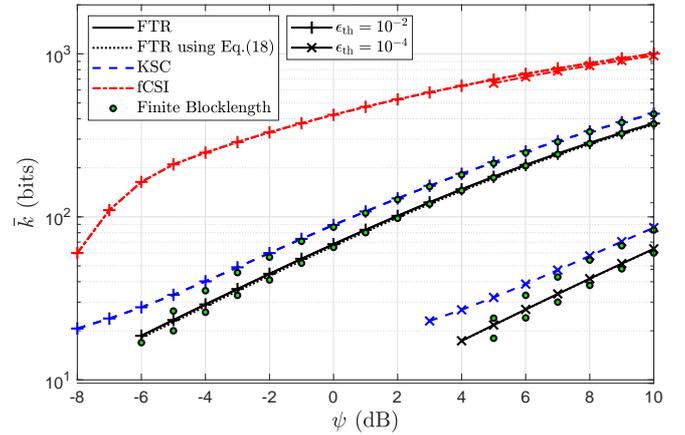}
	\caption{Average message size as a function of $\psi$ when $M=1$, $n=200$ and $v=1000$ channel uses.}
	\label{Fig3}
\end{figure}

Fig.~\ref{Fig3} shows the attainable average message size for each scheme as a function of $\psi$ with $M=1$ and $\epsilon_{\mathrm{th}}\in\{10^{-2},10^{-4}\}$. Obviously, an increase in $\psi$ improves the system performance according to \eqref{psi}; in fact, operating with $\epsilon_{\mathrm{th}}\le 10^{-4}$ is only possible when $\psi\ge 5$ dB. Notice that the schemes KSC and fCSI, that take advantage of certain levels of information, allow to attain greater data rates. In that sense, KSC seems attractive since its performance overcomes significantly the one attained by FTR while making use only of the battery charge information, which is practically viable. However, neither KSC nor fCSI  expand the region for which $\epsilon_{\mathrm{th}}=10^{-4}$ is feasible, below $\psi=5$ dB, and the reasons are highlighted in Section~\ref{KSC}. Also, the greater the reliability constraints, the smaller the message sizes on average. Notice that \eqref{k1G} is accurate, and the results coming from the finite blocklength formulation approximate the asymptotic case for relatively large $\bar{k}$, e.g., for $\bar{k}>50$ bits.

The average message size as a function of the WIT blocklength, $n$, is shown in Fig.~\ref{Fig4} for fixed $\delta=1200$ channel uses while setting $(a)$ $M=1$ and $(b)$ $M=4$. As shown in Fig.~\ref{Fig3}, it is not possible operating with $\epsilon_{\mathrm{th}}=10^{-4}$ when $\psi=0$ dB and $M=1$, thus, Fig.~\ref{Fig4}-$(a)$ shows only the performance for $\epsilon_{\mathrm{th}}=10^{-2}$, while for $M=4$ both reliability constraints are possible to satisfy as shown in Fig.~\ref{Fig4}-$(b)$. From both, Fig.~\ref{Fig4}-$(a)$ and $(b)$, we can appreciate the existence of an optimal point, $(n,v)$, which was analytically proved for the FTR scheme, and notice that $\arg \max\limits_n \bar{k}_{_\mathrm{FTR}}$ matches \eqref{noptFTR} even when using the approximation in \eqref{k1G}, thus, validating those expressions. As pointed out in Section~\ref{FTR} for the FTR scheme, $n^*$ is an increasing function of $\epsilon_{\mathrm{th}}$ and that is why when $M=4$
for $\epsilon_{\mathrm{th}}=10^{-4}\rightarrow n^*=272$ channel uses, while for $\epsilon_{\mathrm{th}}=10^{-2}$ the value of $n^*$ increases to $378$ channel uses. This fact holds also for the KSC and fCSI schemes. Additionally, notice that $n^*\big|_{\mathrm{fCSI}}>n^*\big|_{\mathrm{KSC}}>n^*\big|_{\mathrm{FTR}}$, and KSC and fCSI seem to benefit more of $n>v$ setups, which is an interesting result. 
\begin{figure}[t!]
	\centering
	\includegraphics[width=0.47\textwidth]{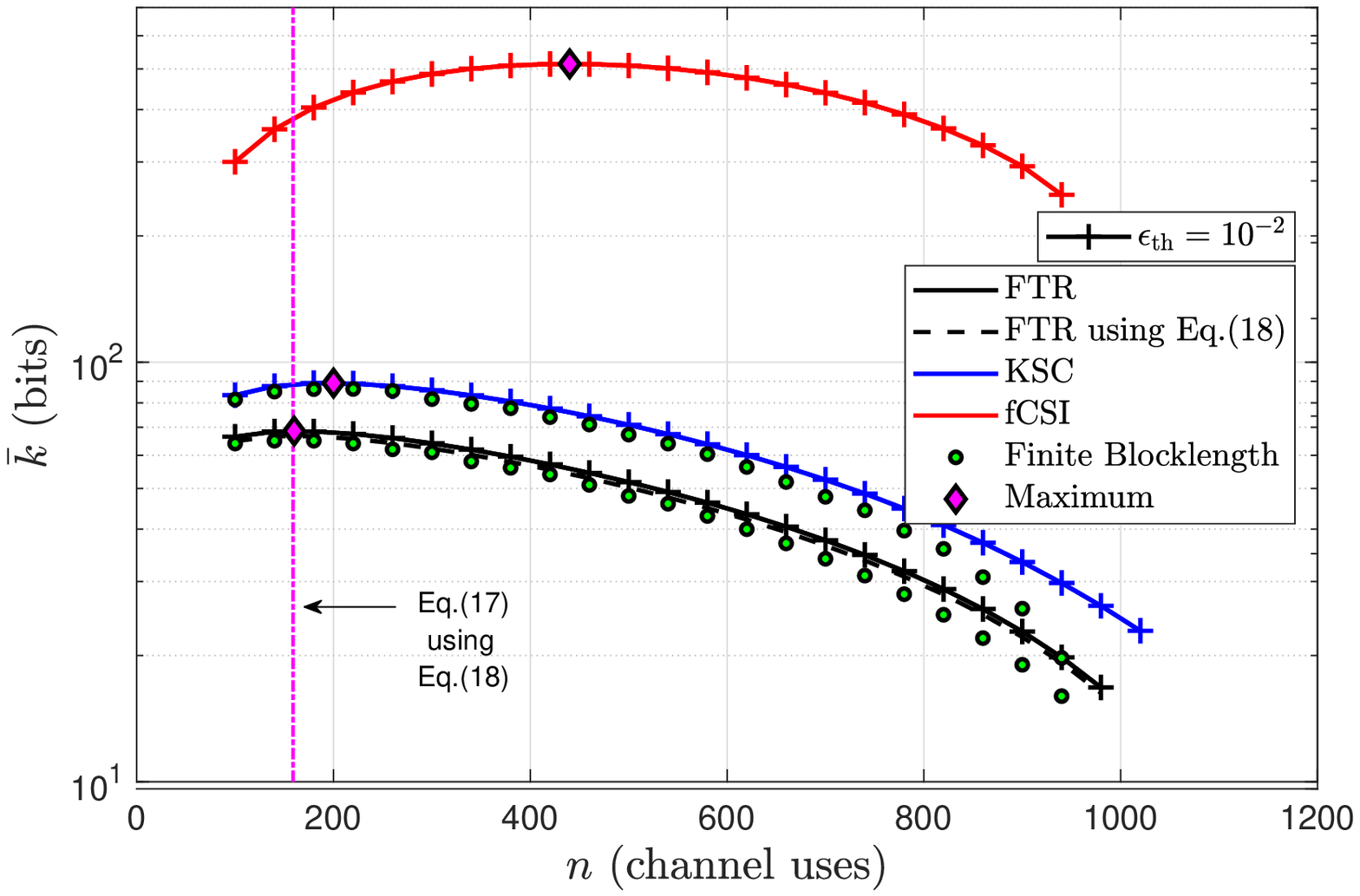}\\
	\includegraphics[width=0.47\textwidth]{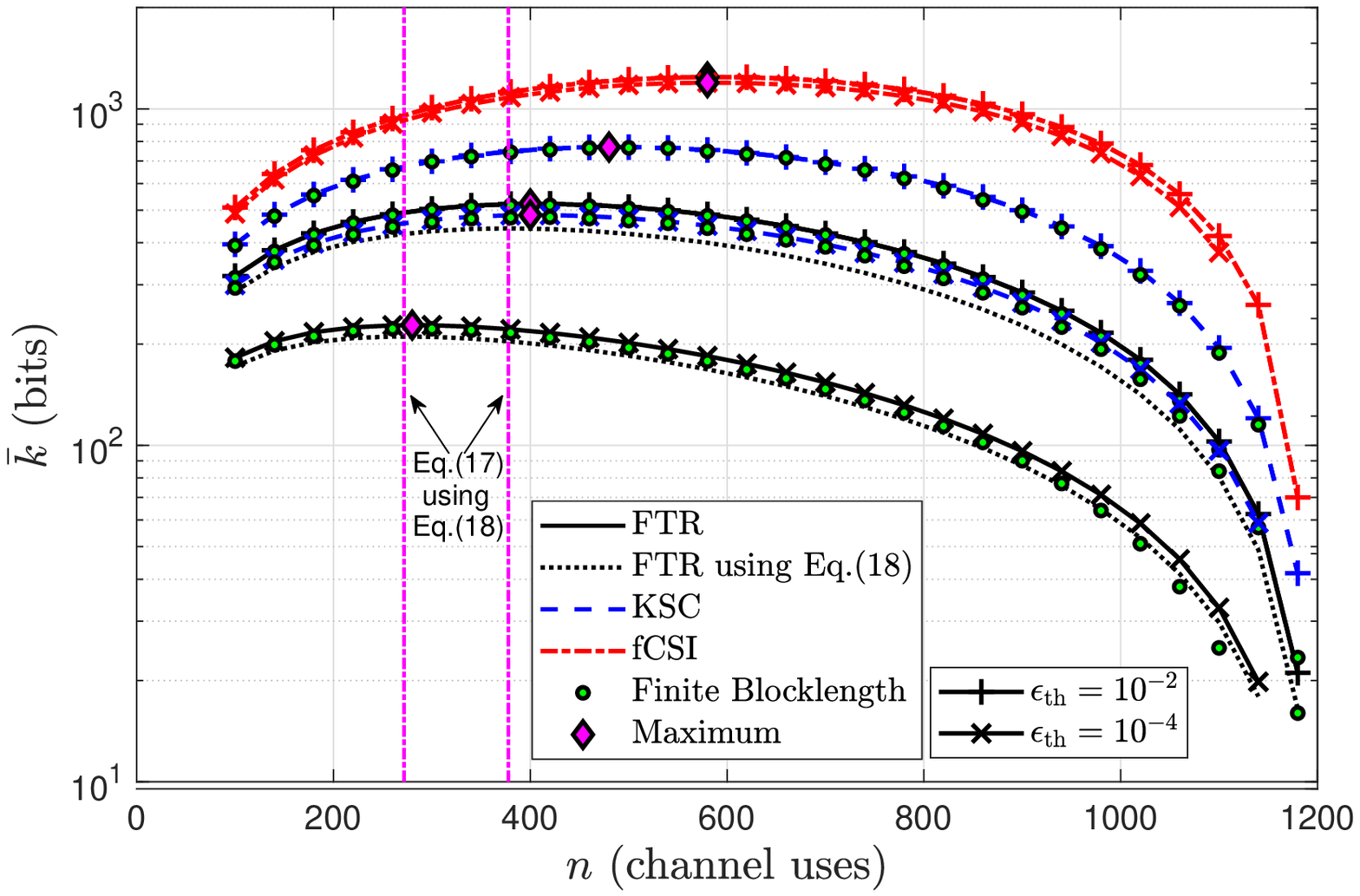}	
	\caption{Average message size as a function of $n$ for $\delta=1200$ channel uses and $(a)\ M=1$ (top),  $(b)\ M=4$ (bottom).}
	\label{Fig4}
\end{figure}
\begin{figure}[t!]
	\centering
	\includegraphics[width=0.47\textwidth]{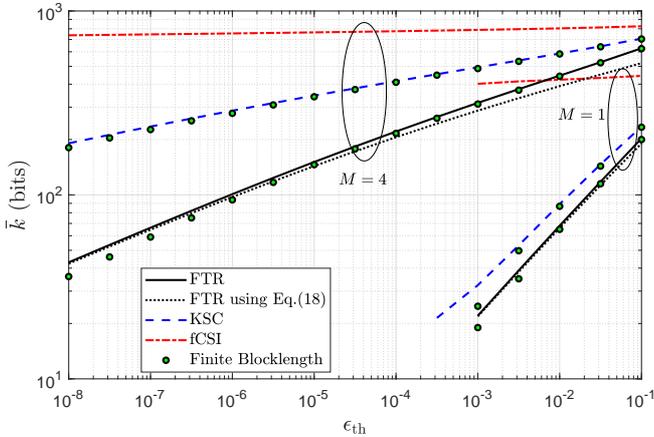}
	\caption{Average message size as a function of $\epsilon_{\mathrm{th}}$ for $M\in\{1,4\}$, $n=200$ and $v=1000$ channel uses.}		
	\label{Fig5}
\end{figure}

\begin{figure}[t!]
	\centering
	\includegraphics[width=0.47\textwidth]{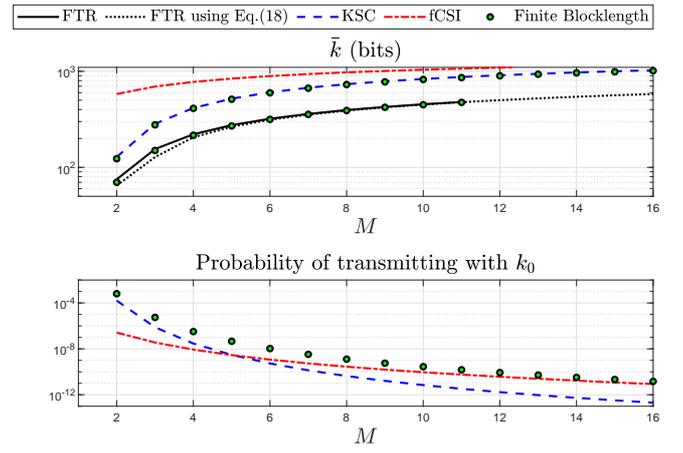}
	\caption{Average message size (top) and probability of transmitting with $k_0$ (bottom), as a function of $M$ for $\epsilon_{\mathrm{th}}=10^{-4}$, $n=200$ and $v=1000$ channel uses.}		
	\label{Fig6}
\end{figure}
\begin{figure}[t!]
	\centering
	\includegraphics[width=0.47\textwidth]{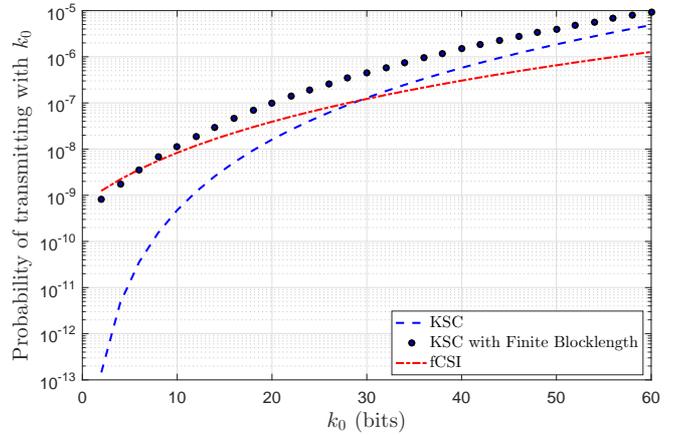}
	\caption{Probability of transmitting with $k_0$, as a function of $k_0$ for $M=4$, $\epsilon_{\mathrm{th}}=10^{-3}$, $n=200$ and $v=1000$ channel uses.}		
	\label{Fig7}
\end{figure}
Fig.~\ref{Fig5} shows the attainable average message size for each scheme as a function of $\epsilon_{\mathrm{th}}$ for $M\in\{1,4\}$. Notice that $\bar{k}$ has a near-linear behavior with respect to $\epsilon_{\mathrm{th}}$ for all the schemes in the log-log scale, thus, $\bar{k}\approx \theta \epsilon_{\mathrm{th}}^{\beta}$, where $\theta$ and $\beta$ depend on the system parameter values, but as shown in the figure they are strongly influenced by $M$, thus, also by $m_1$ and $m_2$. For the FTR scheme this is expected also from \eqref{k1G}. Stringent reliabilities, in the order of $1-10^{-4}$ and greater, are available when using four antennas.
More on this is discussed in Fig.~\ref{Fig6} where we show the positive impact of a greater number of antennas at $D$ for $\epsilon_{\mathrm{th}}=10^{-4}$. 
In fact, the system cannot operate under the reliability requirement when $M=1$, and notice that by increasing $M$ beyond $M=2$, the average message size increases while meeting the required reliability. As $M\rightarrow\infty$, the performance of KSC and fCSI tends to overlap since the resultant WIT channel behaves as an AWGN channel and the randomness in the WIT phase disappears. It is remarkable the performance gain that results from using KSC and fCSI schemes with respect to FTR. For example, for $M=2$, $S$ transmits with $k=70$ bits on average when using FTR, while if the information of the battery state of charge is used (KSC scheme), it can operate with up to $k=123$ bits on average. Additionally, the accuracy of \eqref{k1G} is verified again and its significance is shown to be even greater since for $M\ge 12$ the exact inversion of $F_W(w)$ in \eqref{FW} was not possible to be performed even numerically.
In Fig.~\ref{Fig6} it is also shown the probability of those events in which the allocated message size is the minimum possible, e.g., probability of transmitting with $k_0$, which is desired to be as small as possible.
Since the FTR scheme uses a deterministic rate, we only plot the performance of the KSC and fCSI schemes, for which $\mathbb{P}[k=k_0]=\mathbb{P}[h\le h_0]=F_H(h_0)$ and $\mathbb{P}[k=k_0]=\mathbb{P}[w\le w_0]=F_W(w_0)$, respectively.\footnote{Notice that $\mathbb{P}[k=k_0]>0$ holds always unless the channel becomes deterministic, e.g., $m_1,m_2M\rightarrow\infty$, for which there are only two possibilities: either $\mathbb{P}[k=k_0]=0$  or $\mathbb{P}[k=k_0]=1$.} 
As $M$ increases, the chances of transmitting with the minimum message size decreases, and  the average message size increases. Notice that the gap between the finite blocklength and the asymptotic formulations is considerable when analyzing the probability of transmitting with minimum message size, specially when $M$ increases\footnote{This is because the resultant WIT channel tends to behave as an AWGN channel for which the impact of the finite blocklength is unquestionable as shown in \cite{Polyanskiy.2010}.}, but it does not have a practical impact as shown when evaluating the average message size.

Figs.~\ref{Fig3}-\ref{Fig6} corroborate the appropriateness of using 
the asymptotic results as starting points when getting their finite blocklength counterparts. Finally, Fig.~\ref{Fig7} shows the probability that $S$ transmits with the minimum message size as a function of that size, for $\epsilon_{\mathrm{th}}=10^{-3}$ and $M=4$. As the minimum message size increases, the chances of transmitting messages with that size increases. Differently from the FTR scheme, $\mathbb{P}[k=k_0]$ will never reach the unity when using the KSC\footnote{Unless $\lim\limits_{h_0\rightarrow\infty}z_1(h_0)=\epsilon_{\mathrm{th}}$ for the particular $k_0$ in \eqref{eq}.} and fCSI schemes, thus, $\bar{k}_{_\mathrm{fCSI}}>\bar{k}_{_\mathrm{KSC}}>\bar{k}_{_\mathrm{FTR}}\ge k_0$.
\section{Conclusion}\label{conclusions}
In this paper, we proposed two rate control strategies in order to comply with the reliability and delay constraints of an uMTC WPCN, where multiple antennas are available at the information receiving side. We show the performance gains, with respect to a fixed rate transmission scheme, when the instantaneous battery charge information is used to adapt the transmit rate, while we have also analyzed the ideal scheme requiring full CSI at transmitter side as a benchmark.  The greater the reliability constraints, the smaller the message sizes on average, and the smaller the optimum WIT blocklength. We attain a closed-form expression for the optimum blocklengths when the allocated transmit rate is fixed, and we show that incorporating more information to the rate control schemes, e.g., variable rate schemes, increases the optimum WIT blocklength. Results also show the positive impact of a greater number of antennas, and corroborate the appropriateness of our procedures when using the asymptotic formulation as an approximation of the real finite blocklength results.
\appendices 
\section{Proof of Lemma~\ref{lem_1}}\label{App_A}
We proceed from the PDF of the product of two independent random variables as follows
\begin{align}
f_W(w)&=\int_{0}^{\infty}\frac{1}{\textsl{g}}f_\mathcal{G}(\textsl{g})f_H(w/\textsl{g})\mathrm{d}\textsl{g}\nonumber\\
&=\frac{(m_2M)^{m_2M}m_1^{m_1}w^{m_1-1}}{\Gamma(m_2 M)\Gamma(m_1)}\times\nonumber\\
&\qquad\qquad\times\int_{0}^{\infty}\textsl{g}^{m_2M-m_1-1}e^{-m_2M\textsl{g}-m_1w/\textsl{g}}\mathrm{d}\textsl{g}\nonumber\\
&\stackrel{(a)}{=}\frac{(m_2M)^{m_2M}m_1^{m_1}w^{m_1-1}}{\Gamma(m_2 M)\Gamma(m_1)}\bigg[2\Big(\frac{m_1 w}{m_2M}\Big)^{\frac{m_2M-m_1}{2}}\times\nonumber\\
&\qquad\qquad\times\operatorname{K}_{m_2M-m_1}\big(2\sqrt{m_1m_2M w}\big)\bigg]
\end{align}
\noindent for $w\ge 0$, where $\operatorname{K}_{M-1}(\mathrel{\cdot})$ is the modified Bessel function of second kind and order $M-1$. Step $(a)$ comes from solving the integral using \cite[Eq. (3.471.9)]{Gradshteyn.2014}, then, after simple simplifications we attain \eqref{fW}. Now, using \eqref{fW} the CDF for $w\ge 0$ is attained as follows
\begin{align}
F_W(w)&\!=\!\int_{0}^{w}f_W(u)\mathrm{d}u\nonumber\\
&=\frac{2(m_1m_2M)^{\frac{m_2M+m_1}{2}}}{\Gamma(m_2M)\Gamma(m_1)}\int_{0}^{w}u^{\frac{m_2M+m_1}{2}-1}\times\nonumber\\
&\qquad\qquad\times\operatorname{K}_{m_2M-m_1}\big(2\sqrt{m_1m_2Mu}\big)\mathrm{d}u,
\end{align}
and \eqref{FW} follows immediately after substituting $x=\sqrt{u/w}\rightarrow \mathrm{d}u=2wx\mathrm{d}x$. Unfortunately there is not known closed-form solution for the general case where $m_1,m_2>0$, but specifically for Rayleigh fading scenarios where $m_1=m_2=1$ we can proceed from \eqref{FW} as follows
\begin{align}
	F_W(w)&=\frac{4\big(Mw)^{\frac{M+1}{2}}}{(M-1)!}\int_{0}^{1}x^{M}\operatorname{K}_{M-1}\big(2\sqrt{Mw}x\big)\mathrm{d}x\nonumber\\
	&=\!\frac{4\big(Mw)^{\frac{M\!+\!1}{2}}}{(M\!-\!1)!}\!\bigg[\!\frac{2^{M\!-\!1}\Gamma(M)}{\big(2\sqrt{Mw}\big)^{M\!+\!1}}\!-\!\frac{\operatorname{K}_{\!M}\!\!\big(2\sqrt{Mw}\big)}{2\sqrt{Mw}}\!\bigg]\!,\label{cdf}
\end{align}
where the last step comes from solving the integral using \cite[Eq. (6.561.8)]{Gradshteyn.2014} and after some simple algebraic transformations we reach \eqref{FW1}.\hfill 	\qedsymbol
\section{Proof of Theorem~\ref{the_1}}\label{App_B}
We are interested in the region of high reliability where $w$ and $F_{W}(w)$ are small, thus, according to \cite[Eq. (10.30.2)]{Thompson.2011} and $\operatorname{K}_{v}(q)=\operatorname{K}_{-v}(q)$ we can state
	\begin{align}\label{A1}
	&\operatorname{K}_{m_2M-m_1}(2\sqrt{m_1m_2Mw}x)\nonumber\\
	&\stackrel{w\rightarrow 0}{=} \frac{1}{2}\Gamma\big(|m_2M\!-\!m_1|\big)\big(m_1m_2Mw\big)^{-\!\frac{|m_2M\!-\!m_1|}{2}}x^{-\!|m_2M\!-\!m_1|},
	\end{align}
	which is valid as long as $m_2M\ne m_1$.
	Substituting \eqref{A1} into \eqref{FW} while considering the cases for which $m_2M\gtrless m_1$ and combining them into only one expression, yields
	\begin{align}\label{ap1}
	F_W(w) &\stackrel{w\rightarrow 0}{=} \frac{2\Gamma\big(|m_2M-m_1|\big)}{\Gamma(m_2M)\Gamma(m_1)}\big(m_1m_2Mw\big)^{\min(m_2M,m_1)}\times\nonumber\\
	&\qquad\qquad\qquad\times\int_{0}^{1}x^{2\min(m_2M,m_1)-1}\mathrm{d}x\nonumber\\
	&\stackrel{(a)}{\approx}\! \frac{\Gamma\big(|m_2M\!-\!m_1|\big)\big(m_1m_2Mw\big)^{\min(m_2M,m_1)}}{\Gamma(m_2M)\Gamma(m_1)\min(m_2M,m_1)},
	\end{align}
\noindent where $(a)$ comes from solving the integral. In Fig.~\ref{Fig8}-$(a)$ it is shown that both, the exact and the approximation given in \eqref{ap1} for small $w$, converge in the left tail.
	\begin{figure}[t!]
	\centering
	\includegraphics[width=0.47\textwidth]{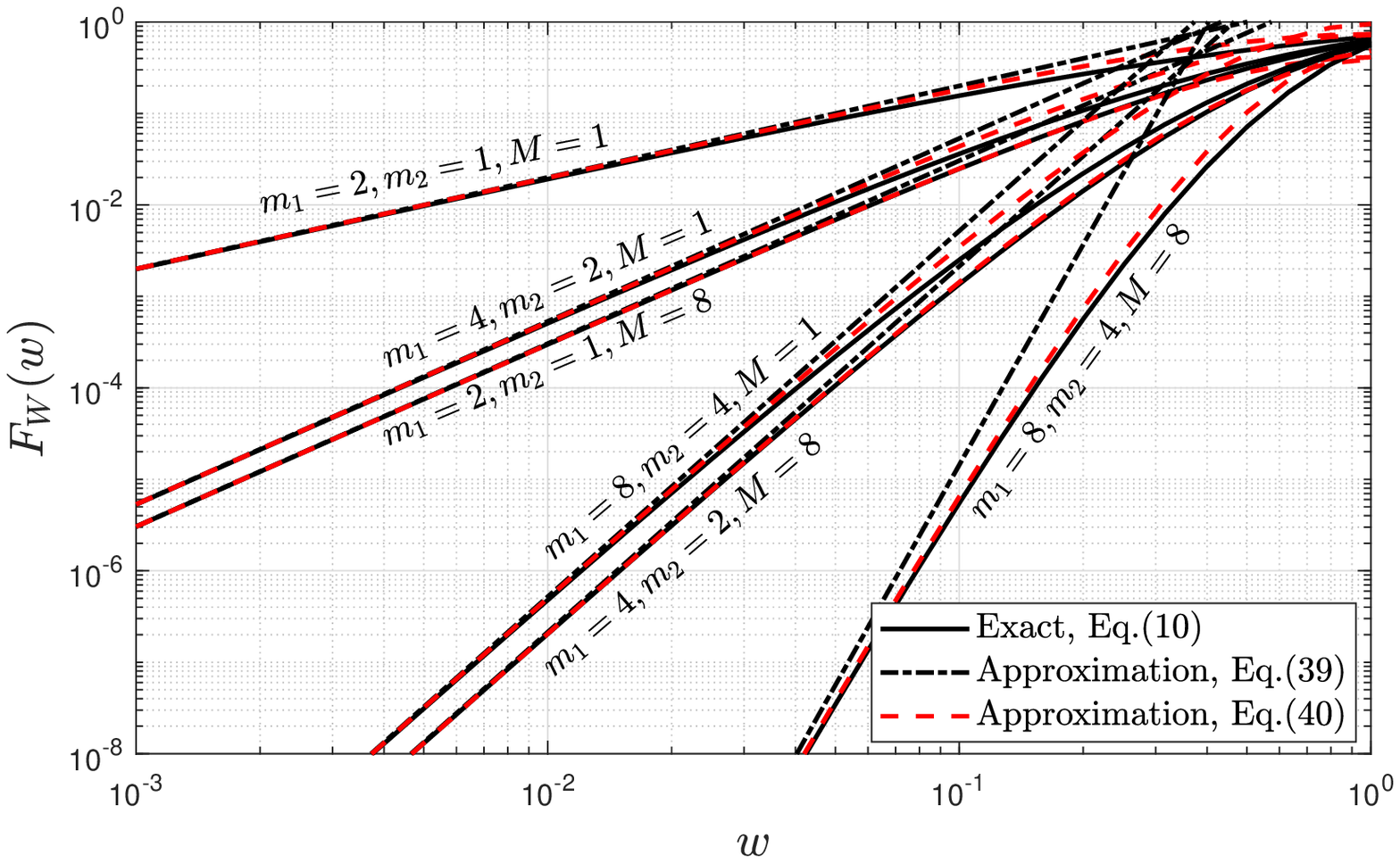}\\
	\includegraphics[width=0.47\textwidth]{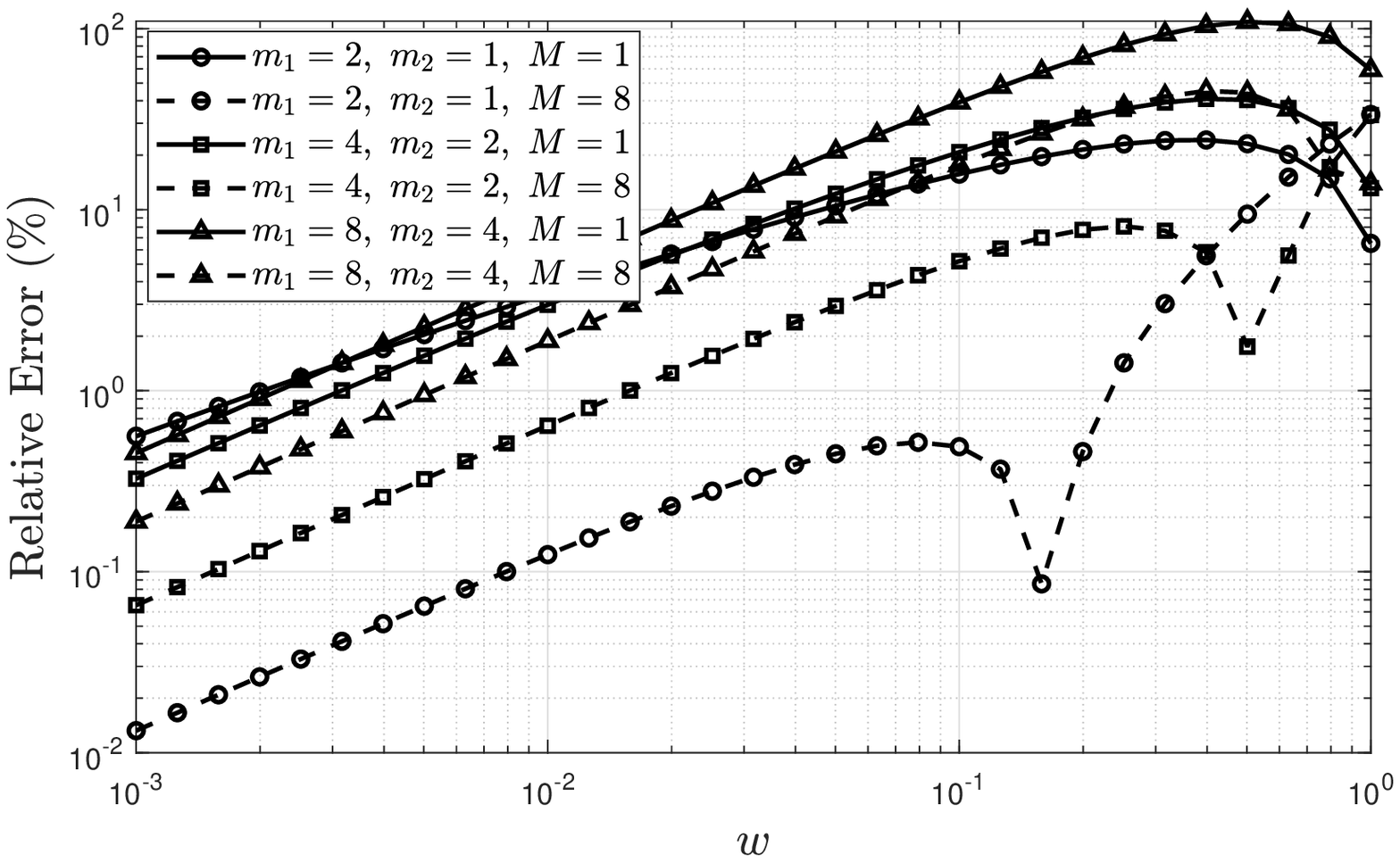}
	\caption{$(a)$ Exact and approximate expressions of $F_{W}(w)$ (top), $(b)$ Relative Error of approximating \eqref{FW} by \eqref{ap2} (bottom).}	
	\label{Fig8}
\end{figure}
 Although \eqref{ap1} is very accurate in the region $F_W(w)<10^{-1}$ for some setups (those with small $\min(m_2M,m_1)$, which is the slope of the curve in the log-log scale according to \eqref{ap1}), there are some others (those with relatively large $\min(m_2M,m_1)$) for which \eqref{ap1} is not very accurate unless for $w\rightarrow 0$. Taking advantage of the influence of the term $\min(m_2M,m_1)$ in that non-linear behavior in the log-log scale, we introduce the factor $e^{-\min(m_2M,m_1)w}$ into \eqref{ap1} yielding \eqref{ap2} for a better fitting.\footnote{The reason behind choosing the exponential function for introducing the non-linearity in the log-log scale is that this kind of function appears in the asymptotic behavior of the modified Bessel functions of second kind \cite[Sec. 10.40]{Thompson.2011}.}  Fig.~\ref{Fig8}-$(a)$ illustrates this, while in Fig.~\ref{Fig8}-$(b)$ it is shown that the relative error\footnote{Notice that the occurrence of local minimums is due to the empiric inclusion of the term $e^{-\min(m_2M,m_1)w}$ in \eqref{ap2}.} decreases quickly as $w\rightarrow 0$, specially for relatively large $M$.
\begin{align}\label{ap2}
	F_W(w) &\approx \frac{\Gamma\big(|m_2M-m_1|\big)\big(m_1m_2Mw\big)^{\min(m_2M,m_1)}}{\Gamma(m_2M)\Gamma(m_1)\min(m_2M,m_1)}\times\nonumber\\
	&\qquad\qquad\qquad\times e^{-\min(m_2M,m_1)w}.
\end{align}	
Now, let us take $\varphi=\min(m_2M,m_1)$ and $\vartheta=\frac{\Gamma\big(|m_2M-m_1|\big)(m_1m_2M)^\varphi}{\varphi \Gamma(m_2M)\Gamma(m_1)}$ and we find $F_{W}^{-1}(\epsilon_{\mathrm{th}})$ departing from \eqref{ap2} as follows
	\begin{align}\label{w1}
	\vartheta w^\varphi e^{-\varphi w}&=\epsilon_{\mathrm{th}}\nonumber\\
	-w e^{-w}&\stackrel{(a)}{=}-\Big(\frac{\epsilon_{\mathrm{th}}}{\vartheta}\Big)^{1/\varphi}\nonumber\\
	w &\stackrel{(b)}{=}-\mathcal{W}\bigg(\!\!-\!\Big(\frac{\epsilon_{\mathrm{th}}}{\vartheta}\Big)^{1/\varphi}\bigg),
	\end{align}
	where $(a)$ comes after some algebraic manipulations that include passing $\vartheta$ to the right term and taking power $1/\varphi$ in both sides, while $(b)$ follows from using the definition of the Lambert $W$ function, specifically its main branch since $w\rightarrow 0$ and $-w>-1$, which guarantees finding the appropriate real solution for the equation. Finally, substituting the expressions of $\vartheta$ and $\varphi$ into \eqref{w1} we reach \eqref{k1G}. \hfill
\qedsymbol
\section{Proof of Theorem~\ref{the_2}}\label{App_C}
Since $\delta=v+n$, we can put \eqref{k1} in the form of $k_{_{\mathrm{FTR}}}\approx n\log_2(1+(\delta-n)\chi/n)$, and its first and second derivatives in terms of $n$ are given next.
\begin{align}
\frac{d k_{_{\mathrm{FTR}}}}{d n}&\!\approx\!\bigg[\frac{-\chi\delta}{\chi(\delta\!-\!n)\!+\!n}\!+\!\ln\Big(1\!+\!\frac{\delta\!-\!n}{n}\chi\Big)\bigg]\!\log_2e,\label{d1}\\
\frac{d^2 k_{_{\mathrm{FTR}}}}{d n^2}&\!\approx-\frac{\chi^2\delta^2\log_2e}{n\big(\chi(\delta-n)+n\big)^2}.\label{d2}
\end{align}
Now, $k_{_{\mathrm{FTR}}}$ is concave on $n$ since $\frac{d^2 k_{_{\mathrm{FTR}}}}{d n^2}<0$. Based on this, and since for $n=0$ and $n=\delta$ we have that $k_{_\mathrm{FTR}}=0$, we conclude that there is an optimum $0<n^*<\delta$ and it comes from solving $\frac{d k_{_{\mathrm{FTR}}}}{d n}=0$. Using \eqref{d1} we proceed as next
\begin{align}
\frac{-\chi\delta}{\chi(\delta-n)+n}+\ln\Big(1+\frac{\delta-n}{n}\chi\Big)&=0& \nonumber\\
\frac{1-\chi-y(n)}{y(n)}+\ln y(n)&\stackrel{(a)}{=}0\nonumber\\
y(n)\big[\ln y(n)-1\big]&\stackrel{(b)}{=}\chi-1\nonumber\\
 e^{\ln y(n)-1}\big[\ln y(n)-1\big]&\stackrel{(c)}{=}\frac{\chi-1}{e}\nonumber\\
\ln y(n)-1&\stackrel{(d)}{=}\mathcal{W}\Big(\frac{\chi-1}{e}\Big)\nonumber\\ y(n)&\stackrel{(e)}{=}\frac{\chi-1}{\mathcal{W}\Big(\frac{\chi-1}{e}\Big)},
\end{align}
where $(a)$ comes from using $y(n)=1+\frac{\delta-n}{n}\chi$, $(b)$ follows from multiplying by $y(n)$ on each side and rearranging the equation, $(c)$ from using $y(n)=e^{\ln y(n)}$ and some simple algebraic transformations, $(d)$ follows immediately from the definition of the Lambert $W$ function, which is taken in the main branch since $\ln y(n)-1>-1$, and $(e)$ comes from isolating $y(n)$ and using $\mathcal{W}(x)e^{\mathcal{W}(x)}=x\rightarrow e^{\mathcal{W}(x)}=x/\mathcal{W}(x)$. Then, \eqref{noptFTR} is reached straightforward by using $n=\frac{\delta\chi}{y(n)+\chi-1}$ along with the fact that $n$ must be a positive integer. The optimum $v$ is easily obtained by using $n^*+v^*=\delta$. \hfill 	\qedsymbol
\section{Proof of Theorem~\ref{the3}}\label{App_F}
Let's define $z(h_0)=z_1(h_0)+z_2(h_0)$, where $z_1$ and $z_2$ are given in \eqref{eq}, and $z(0)=1>\epsilon_{\mathrm{th}}$, then
\begin{align}
\frac{d z(h_0)}{d h_0}&\!=\frac{d z_1(h_0)}{d h_0}+\frac{d z_2(h_0)}{d h_0}\nonumber\\
&=\frac{m_1^{m_1}h^{m_1-1}e^{-m_1h}}{\Gamma(m_1)}\bar{\epsilon}(k_0,n,h)\Big|_{h_0}+\frac{d \bar{\epsilon}(k_0,n,h_0)}{d h_0}\times\nonumber\\
&\ \  \times\frac{\Gamma(m_1,m_1h_0)}{\Gamma(m_1)}-\frac{m_1^{m_1}h_0^{m_1-1}e^{-m_1h_0}}{\Gamma(m_1)}\bar{\epsilon}(k_0,n,h_0)\nonumber\\
&=\frac{d \bar{\epsilon}(k_0,n,h_0)}{d h_0}\frac{\Gamma(m_1,m_1h_0)}{\Gamma(m_1)}\nonumber\\
&\stackrel{(a)}{=}\frac{\Gamma(m_1,m_1h_0)}{\Gamma(m_1)}\frac{d}{d h_0}F_{\mathcal{G}}\Big(\frac{\kappa}{Mh_0}\Big)\nonumber\\
&\!\stackrel{(b)}{=}\!\frac{-\Gamma(m_1,m_1h_0)}{\Gamma(m_1)\Gamma(m_2M)}\big(\kappa m_2\big)^{m_2M}h_0^{\!-\!m_2M\!-\!1}\!e^{\!-\!\frac{\kappa m_2}{h_0} }\!<\!0,\label{der}
\end{align}
where $(a)$ comes from using \eqref{eqq} with $\frac{(2^{k/n}-1)}{\psi }\frac{n}{v}=\kappa$, while $(b)$ follows from using the expression of the CDF of $\mathcal{G}$ and taking the derivative. Because $d z(h_0)/d h_0<0$, $z(h_0)$ is strictly decreasing and therefore it will intersect with the horizontal line $\epsilon_{\mathrm{th}}$ in a unique point, thus, the solution of \eqref{eq} is unique and any bracketing root finding method, e.g., bisection method, converges to it. However, for greater convergence rates, other more evolved techniques, such as the Newton's method, are recommended. Let's define $\tilde{z}(h_0)=z(h_0)-\epsilon_{\mathrm{th}}$ while we need to solve $\tilde{z}(h_0)=0$. Notice that $d\tilde{z}(h_0)/dh_0=dz(h_0)/dh_0$, and as shown in \eqref{der} it can be efficiently computed. Thus, using the recursive relation of Newton's method, given by
\begin{align}\label{Newton}
h_0^{(t+1)}=h_0^{(t)}-\frac{\tilde{z}(h_0^{(t)})}{d \tilde{z}(h_0^{(t)})/d h_0^{(t)}},
\end{align}
does not require much more computational effort.\footnote{Notice that, $\tilde{z}(x)$ is more difficult to compute than $d\tilde{z}(x)/d x$, thus, compared with bracketing methods that make use of $\tilde{z}(x)$, Newton's method involves almost the same computational effort in this case.} Parameter $t$ indicates the iteration index, and it is known that Newton's method always converges if \footnote{This result comes from Fixed Point Theory \cite{Agarwal.2001}.}
\begin{align}\label{nmu}
\mu=\Big|\frac{\tilde{z}(h_0)d^2\tilde{z}(h_0)/dh_0^2}{d\tilde{z}(h_0)/dh_0}\Big|\le K<1
\end{align} 
in the searching region with $K$ as a positive constant. Now, a suitable root finding method that uses a combination of Newton's and bisection methods for our problem setup, is presented in \textbf{Algorithm~\ref{alg1}}. 

To guarantee that the initial guess of $h_0$, $h_0^{(1)}$, satisfies \eqref{nmu}, we chose it such that 
\begin{algorithm}[t!]
	\caption{Finding $h_0$ in \eqref{eq}}
	\label{alg1}		
	\begin{algorithmic} [1]
		\State $h_0^{(0)}=0$, $h_0^{(1)}=\rho$ in \eqref{rho}, $t=1$   \label{l1}
		\While{$100\times(h_0^{(t)}-h_0^{(t-1)})/h_0^{(t-1)}>h_{_\triangle}$}   \label{l2}
		\If{$\tilde{z}(h_0^{(t)})/\tilde{z}(h_0^{(t-1)})>0$}	       \label{l3}
		\State Calculate $h_0^{(t+1)}$ according to \eqref{Newton}     \label{l4}		
		\Else
		\State Calculate $\mu$ given in \eqref{nmu}                    \label{l6}
		\If{$\mu< 1$}                                                \label{l7}
		\State Calculate $h_0^{(t+1)}$ according to \eqref{Newton}     \label{l8}		
		\Else                      
		\State $h_0{(t+1)}=(h_0^{(t)}+h_0^{(t+1)})/2$                  \label{l10}
		\EndIf		
		\EndIf 	
		\State $t\leftarrow t+1$                                       \label{l13}
		\EndWhile		              
		\State End                                                     \label{l15}
	\end{algorithmic}
\end{algorithm}
\begin{align}
\frac{d^2\tilde{z}(h_0)/dh_0^2}{d\tilde{z}(h_0)/dh_0}&=0\nonumber\\
-\frac{h_0-\kappa m_2+h_0 m_2M+\frac{m_1^{m_1}e^{-m_1h_0}h_0^{m_1+1}}{\Gamma(m_1,m_1h_0)}}{h_0^2}&\stackrel{(a)}{=}0\nonumber\\
\frac{h_0-\kappa m_2+h_0 m_2M}{h_0^2}+\frac{m_1^{m_1}e^{-m_1h_0}h_0^{m_1-1}}{\Gamma(m_1,m_1h_0)}&=0,\label{eqder}
\end{align}
where $(a)$ comes from taking the derivative of \eqref{der} and computing the quotient $\frac{d^2\tilde{z}(h_0)/dh_0^2}{d\tilde{z}(h_0)/dh_0}$. Using the fact that $h_0\ll 1$ we know that $\Gamma(m_1,m_1h_0)\approx \Gamma(m_1)$, and if additionally $m_1>1$ then we can ignore the impact of the second summand in \eqref{eqder}, thus, 
\begin{align}
h_0-\kappa m_2+h_0 m_2M=0\Rightarrow h_0\approx \frac{\kappa m_2}{1+m_2 M}.\label{h0eq}
\end{align}
Notice that $m_1>1$ is a usual characteristic of WPCNs where the WET process requires a considerable LOS component. Thus, \eqref{h0eq} holds for practical setups.
For the sake of completeness we also present the solution for $m_1=1$, which follows after \eqref{eqder} as shown next\footnote{For setups with $m_1<1$ we were not able of reaching an accurate approximation, however, those scenarios are not of practical interest.}
\begin{align}
\frac{h_0-\kappa m_2+h_0 m_2M}{h_0^2}+\frac{e^{-h_0}}{\Gamma(1,h_0)}&=0\nonumber\\
h_0^2+(m_2M+1)-\kappa m_2&\stackrel{(a)}{=}0\nonumber\\
\sqrt{\Big(\frac{m_2M+1}{2}\Big)^2+\kappa m_2}-\frac{m_2M+1}{2}&\stackrel{(b)}{=}h_0,\label{sech0}
\end{align}
where $(a)$ comes from simple algebraic manipulations after using $\Gamma(1,h_0)=e^{-h_0}$, while $(b)$ follows after solving the resulting quadratic equation and taking the positive real solution. Combining the results in \eqref{h0eq} and \eqref{sech0}, the initial guess for $h_0$ is set to
\begin{equation}\label{rho}
\rho=\left\{ \begin{array}{lc}
\frac{\kappa m_2}{1+m_2M}, &\ \mathrm{if}\   m_1>1 \\
\sqrt{\Big(\frac{m_2M\!+\!1}{2}\Big)^2\!+\!\kappa m_2}\!-\!\frac{m_2M\!+\!1}{2}, &\ \mathrm{if}\ m_1=1
\end{array}
\right.,
\end{equation}

Back to \textbf{Algorithm~\ref{alg1}}, lines \ref{l2}-\ref{l13} define the iteration cycle that stops when the required accuracy in the solution, $h_{_\triangle}$ given in percent, is attained. If $h_0^{(t)}$ evolves approaching $h_0$ without over crossing it, then it is secure to use Newton's iteration as shown in lines \ref{l3} and \ref{l4} of the algorithm. Otherwise we have to check the convergence property given in \eqref{nmu} to be sure that the next Newton's iteration will be closer to the solution. If the condition holds we use \eqref{Newton} (line \ref{l8}), otherwise it is better to average the outputs of the two latest iterations since they are bracketing the solution (line \ref{l10}), e.g., bisection rule. The global convergence is ensured. In fact, it is expected that \textbf{Algorithm~\ref{alg1}} always uses  Newton's rule at each iteration, and the bisection rule would be required in some odd cases only.

Now, for $k>k_0$ we know that $\bar{\epsilon}(k,n,h)=\epsilon_{\mathrm{th}}^*=\bar{\epsilon}(k_0,n,h_0)$, thus, based on \eqref{eqq} we have that
\begin{align}
F_{\mathcal{G}}\bigg(\frac{(2^{k/n}-1)n}{\psi v M h}\bigg)&=F_{\mathcal{G}}\bigg(\frac{(2^{k_0/n}-1)n}{\psi v M h_0}\bigg)\nonumber\\
\frac{2^{k/n}-1}{h}&=\frac{2^{k_0/n}-1}{h_0},
\end{align}
and isolating $k$ we attain \eqref{kexp}.

\bibliographystyle{IEEEtran}
\bibliography{IEEEabrv,references}
\end{document}